\documentclass[a4paper,12pt]{article}

\usepackage{complexity} 
\usepackage{enumerate}
\usepackage{amsmath, amssymb, amsthm,complexity}
\usepackage{euler}
\usepackage{authblk}
\usepackage{todonotes}

\usepackage{geometry}
\usepackage{mathtools}
\usepackage{thmtools}
\usepackage{thm-restate}
\usepackage{hyperref}
\usepackage{cleveref}
\usepackage{complexity}
\usepackage{authblk}
\theoremstyle{plain}
\usepackage{todonotes}
\usepackage{microtype}
\usepackage{multirow}
\usepackage{amsmath,amssymb}
\usepackage{comment}
\usepackage{enumerate}
\usepackage{xspace}
\usepackage{listings}
\usepackage{color}
\usepackage{cite}
\usepackage{graphicx}
\RequirePackage{fancyhdr}
 \usepackage{xcolor}
\usepackage{boxedminipage}
\usepackage{algorithm2e}

\newcommand{\defparproblem}[4]{
  \vspace{3mm}
\noindent\fbox{
  \begin{minipage}{.95\textwidth}
  \begin{tabular*}{\textwidth}{@{\extracolsep{\fill}}lr} \textsc{#1}\\ \end{tabular*}
  {\bf{Input:}} #2  \\
  {\bf{Parameter:}} #3 \\
  {\bf{Question:}} #4
  \end{minipage}
  }
  \vspace{2mm}
}

\newcommand{\defproblem}[3]{
  \vspace{3mm}
\noindent\fbox{
  \begin{minipage}{.95\textwidth}
  \begin{tabular*}{\textwidth}{@{\extracolsep{\fill}}lr} #1  \\ \end{tabular*}
  {\bf{Input:}} #2  \\
  {\bf{Question:}} #3
  \end{minipage}
  }
  \vspace{2mm}
  }

\def\mypara#1{\sbox0{\parbox{\linewidth}{%
  \parindent0pt
 #1\par\xdef\myparasize{\the\prevgraf}}}%
\ifnum\myparasize=1
{\parindent0pt #1\par}%
\else
#1
\fi}

\newcommand{\AAA}{{\mathcal A}}

\newcommand{\OO}{{\mathcal O}}

\newcommand{\FF}{{\mathcal F}}

\newcommand{\SSS}{{\mathcal S}}

\newcommand{\GG}{{\mathcal G}}
\newcommand{\pionepitwodeletion}{{\sc $\Pi_1$ or $\Pi_2$ Deletion}}
\newcommand{\fpionepitwodeletion}{{\sc Finite $\Pi_1$ or $\Pi_2$ Deletion}}
\newcommand{\fpionepitwodeletionpath}{{\sc Finite $\Pi_1$ or $\Pi_2$ Deletion with Path}}

\newcommand{\ssp}{{\sf sp}}

\newcommand{\sv}[1]{}

\newtheorem{corollary}{Corollary}
\newtheorem{reduction rule}{Reduction Rule}
\newtheorem{branching rule}{Branching Rule}
\newtheorem{definition}{Definition}

\newtheorem{lemma}{Lemma}
\newtheorem{observation}{Observation}

\newtheorem{theorem}{Theorem}

\usepackage{listings}
\usepackage{color}
\usepackage{cite}
\usepackage{complexity}
\usepackage{todonotes}
\usepackage{enumitem}
\newcommand{\PIVDTrees}{{\sc Proper Interval-or-Tree Deletion}}
\newcommand{\IVDTrees}{{\sc Interval-or-Tree Deletion}}
\newcommand{\ChordalBipPer}{{\sc Chordal-or-Bipartite Permutation Deletion}}

\newcommand{\InfCMSOPiOnePiTwo}{{\sc Special Infinite-$(\Pi_1, \Pi_2)$-Deletion}}
\newcommand{\CliqueorPlanar}{{\sc Cliques-or-Planar Deletion}}
\newcommand{\ClawfreeorTriangleFree}{{\sc Claw-Free-or-Triangle-Free Deletion}}
\newcommand{\CliqueorPiTwo}{{\sc Cliques-or-$\Pi_2$ Deletion}}
\newcommand{\SplitorBipartite}{{\sc Split-or-Bipartite Deletion}}
\newcommand{\CliqueorKtFree}{{\sc Clique-or-$K_t$-free-$\Pi_2$ Deletion}}
\newcommand{\PalphafreePiOnePiTwo}{{\sc $P_\alpha$-Free-$(\Pi_1, \Pi_2)$-Deletion}}

\title{Deletion to Scattered Graph Classes II - Improved FPT Algorithms for Deletion to Pairs of Graph Classes \footnote{Preliminary versions of the paper appeared in proceedings of IPEC 2020 \cite{jacob2020parameterized} and FCT 2021 \cite{JacobMR21fct}.}}

\date{}
\author[1]{Ashwin Jacob}
\author[2]{Diptapriyo Majumdar}
\author[3]{Venkatesh Raman}
\affil[1]{Ben-Gurion University of the Negev, Beersheva, Israel
  \texttt{ashwinj@bgu.ac.il}}
\affil[2]{Indraprastha Institute of Information Technology Delhi, New Delhi, India
    \texttt{diptapriyo@iiitd.ac.in}}
\affil[3]{The Institute of Mathematical Sciences, HBNI, Chennai, India
  \texttt{vraman@imsc.res.in}}

\begin{document}
\maketitle

\begin{abstract}
The problem of deletion of vertices to a hereditary graph class is a well-studied problem in parameterized complexity. Recently, a natural extension of the problem was initiated where we are given a finite set of hereditary graph classes and we determine whether $k$ vertices can be deleted from a given graph so that the connected components of the resulting graph belong to one of the given hereditary graph classes.  The problem is shown to be fixed parameter tractable
(FPT) when the deletion problem to each of the given hereditary graph classes is fixed-parameter tractable, and the property of being in any of the graph classes is expressible in the counting monodic second order (CMSO) logic. 
This paper focuses on pairs of specific graph classes ($\Pi_1,\Pi_2$) in which we would like the connected components of the resulting graph to belong to, and design simpler and more efficient FPT algorithms.

\end{abstract}
\section{Introduction}
Graph modification problems are a class of problems in which the input instance is a graph, and the goal is to check if the input can be transformed into a graph of a specified graph class by using some ``allowed'' graph operations.
Depending on the allowed operations, {\em vertex or edge deletion problems}, {\em edge editing or contraction problems} have been extensively studied in various algorithmic paradigms.

In the last two decades, graph modification problems, specifically vertex deletion problems, have been extensively studied in the field of parameterized complexity. Examples of vertex deletion problems include {\sc Vertex Cover}, {\sc Cluster Vertex Deletion}, {\sc Feedback Vertex Set}, and {\sc Chordal deletion set}. We know from the classical result by Lewis and Yannakakis \cite{LewisY80} that the problem of whether removing a set of at most $k$ vertices results in a graph satisfying a hereditary property $\pi$ is NP-complete for every non-trivial property $\pi$. 
It is well-known that any hereditary graph class\footnote{A hereditary graph class is a class of graphs that is closed under induced subgraphs} can be described by a forbidden set of graphs, finite or infinite, that contains all minimal forbidden graphs in the class. It is also well-known~\cite{Cai96} that if a hereditary graph class has a finite forbidden set, then deletion to the graph class has a simple fixed-parameter tractable (FPT) algorithm using a hitting set based approach.

Recently Jacob et al.~\cite{jacob2020parameterized,jacob2021parameterized}, building on the work of Ganian et al.~\cite{GanianRS17} for constraint satisfaction problems, introduced a natural extension of vertex deletion problems to deletion to scattered graph classes. Here we want to delete vertices from a given graph to put the connected components of the resulting graph to one of a few given graph classes.
A scattered graph class $(\Pi_1, \dotsc ,\Pi_d)$ consists of graphs whose connected components are in one of the graph classes $\Pi_1, \dotsc ,\Pi_d$. The vertex deletion problem to this class cannot be solved by a hitting set based approach, even if the forbidden graphs for these classes are finite.
In particular, it is possible that the solution could be disjoint from the forbidden subgraphs present in the input instance.
It is sufficient if the solution separates a forbidden subgraph from one class from a forbidden subgraph of another class so that the forbidden subgraphs of the $d$ classes don't belong to the same component.

Jacob et al.~\cite{jacob2020parameterized} proved that the vertex deletion problem for the scattered graph class $(\Pi_1, \dotsc , \Pi_d)$ is FPT  with running time $2^{poly(k)} n^{\OO(1)}$ ($poly(k)$ denotes a polynomial in $k$) if the forbidden families corresponding to all the graph classes $\Pi_1, \dotsc , \Pi_d$ are finite. The technique involves iterative compression and important separator variants. In a later result \cite{jacob2021parameterized}, they showed that when the vertex deletion problem to each of the individual graph classes is FPT and for each graph class, the property that a graph belongs to the graph class is expressible by Counting Monadic Second Order (CMSO) logic. Unfortunately, the running time of the algorithm incurs a gargantuan constant factor (a function of $k$) overhead. 


Since the algorithms in \cite{jacob2020parameterized,jacob2021parameterized} incur a huge running time and use sophisticated techniques, it is interesting to see whether we can get simpler and faster algorithms for some special cases of the problem. In this paper, we do a deep dive on the vertex deletion problems to a pair of graph classes when at least one of the graph classes has an infinite forbidden family. 

\noindent
{\bf Our Problems, Results, and Techniques:} We look at specific variants of the following problem.

\defparproblem{{\pionepitwodeletion}}{An undirected graph $G = (V, E)$, two hereditary graph classes $\Pi_1$ and $\Pi_2$ with $\FF_i$ as the forbidden family for graphs whose each connected component belongs to $\Pi_i$ for $i \in \{1,2\}$. }{$k$}{Is there a set $S \subseteq V(G)$ of size at most $k$ such that every connected component of $G - S$ is in $\Pi_1$ or in $\Pi_2$?}

We emphasize that two distinct components of $G-S$ can be in two distinct classes $\Pi_1$ and $\Pi_2$. We do not ask that every component of the resulting graph is in $\Pi_1$ or that every component of the resulting graph is in $\Pi_2$.
Also, note that the forbidden families $\FF_i$ are for graphs whose each connected component is a graph belonging to $\Pi_1$ for $i \in \{1,2\}$. It is not the forbidden family of graphs associated to the graph class $\Pi_i$. This distinction does not make a difference for most of the popular graph classes as the union of connected components of such graph classes still belong to the graph class. Examples include bipartite graphs, chordal graphs, planar graphs, interval graphs, and forests.
However, it is important for classes such as cliques and split graphs. The forbidden family for cliques is the singleton graph $2K_1$. But if the graph class is such that each connected component is a clique, $2K_1$ is present by taking a single vertex from two different components of the graph. The forbidden family in this case can be proven to be the singleton graph $P_3$. In our definition of {\pionepitwodeletion}, when the graph class $\Pi_1$ is the class of cliques $\FF_1 = \{P_3\}$ is the forbidden family of graphs where each component is a clique. 

We describe a general algorithm for {\pionepitwodeletion} under some conditions which cover pairs of several graph classes. While the specific conditions on the pairs of classes to be satisfied by this algorithm are somewhat technical and are explained in Section \ref{subsection-forbidden-paths} and \ref{subsection:general-algorithm}, we give a high-level description here. We note that we do not put any CMSO logic based conditions; thus the algorithm solves fixed parameter tractability for pairs of graph classes not coming under the result by Jacob et al in \cite{jacob2021parameterized}.

We first make the reasonable assumption that the vertex deletion problems to the graph class $\Pi_1$ and to $\Pi_2$ have FPT algorithms.
As we want every connected component of the graph after removing the solution vertices to be in $\Pi_1$ or in $\Pi_2$, any pair of forbidden subgraphs $H_1 \in \FF_1$ and $H_2 \in \FF_2$ cannot both 
 be in a connected component of $G$. Let us look at such a component $C$ with $J_1, J_2 \subseteq V(C)$ such that $G[J_i]$ is isomorphic to $H_i$ for $i \in \{1,2\}$ and look at a path $P$ between the sets $J_1$ and $J_2$. Assuming that the graphs in families $\FF_1$ and $\FF_2$ are connected graphs, we can conclude that the solution has to hit the set $J_1 \cup J_2 \cup P$ allowing a branching on such sets.

However, if the path is too large, such a branching does not lead to efficient algorithms. The generalization comes up from our observation that for certain pairs of graph classes, if we focus on a pair of forbidden subgraphs $H_1 \in \FF_1$ and $H_2 \in \FF_2$ that are ``closest" to each other, then there is always a solution that does not intersect the shortest path $P$ between them. This helps us to branch on the vertex sets of these forbidden graphs. However, note that the forbidden 
graphs may have unbounded sizes. We come up with a notion of {\it forbidden pair} (Definition~\ref{forbpair} in Section~\ref{sec:forbidden-characterization}) and show that there are pairs of graph classes that have a finite number of forbidden pairs even if each of them has infinite forbidden sets. For some such pairs, we can bound the branching step to obtain the FPT algorithm.


\noindent
{\bf Organization of the paper:}
In Section \ref{sec:Prelims}, we state the notations used in this paper and give the necessary preliminaries on various graph classes and parameterized complexity. We also state some preliminary observations and reduction rules for {\pionepitwodeletion}. In Section \ref{sec:finite-set-with-paths}, we give algorithms and kernels for the simplest case of the problem when both of the forbidden families associated with the pair of graph classes are finite, and one of them has a path graph present in it. 

In Section \ref{sec:constant-forbidden-pair}, we give algorithms {\pionepitwodeletion} having a constant sized forbidden pair families. We define the notion of forbidden pair family and associated characterizations in Section \ref{sec:forbidden-characterization}. In Section \ref{subsection-forbidden-paths}, we first give an algorithm for {\pionepitwodeletion} having a constant sized forbidden pair families assuming that one of the families has a path graph present in it. Later, we give some examples of pairs of graph classes satisfying these properties. 

In Section \ref{subsection:constant-forbidden-pair}, we give  
algorithms for {\pionepitwodeletion} having a constant sized forbidden pair families satisfying some additional conditions. For these problems, the path between the closest forbidden pair need not be finite like the examples in Section \ref{sec:constant-forbidden-pair}. We motivate this algorithm with the example of {\ClawfreeorTriangleFree} in Section \ref{section:clae-free-triangle-free}. In Section \ref{subsection:general-algorithm}, we give the general algorithm for this case. In Section \ref{sec:other-examples}, we give more examples for this case, such as {\IVDTrees}, {\PIVDTrees}, and {\ChordalBipPer}.

\section{Preliminaries}
\label{sec:Prelims}
\paragraph{\bf Sets and Graph Theory:}
Given $r \in \mathbb{N}$, we use $[r]$ to denote the set $\{1,\ldots,r\}$.
Given a finite set $A$, and an integer $t$, we use ${{A}\choose{t}}$ to denot the collection of all subsets of $A$ of size exactly $t$ and use ${{A}\choose{\leq t}}$ to denote the collection of all subsets of $A$ of size  at most $t$.
We consider undirected graphs throughout this paper.
We use standard graph-theoretic notations for undirected graphs from \cite{DiestelBook2012}.
For $\ell \in \mathbb{N}$, we use $P_{\ell}$ to denote the path on $\ell$ vertices.
Similarly, for $\ell \in \mathbb{N}$, we use $C_{\ell}$ to denote an induced cycle on $\ell$ vertices.
Let $u, v \in V(G)$ and $K_{\ell}$ to denote a clique with $\ell$ vertices.
We use $d_G(u, v)$ to denote the length of a {\em shortest path} from $u$ to $v$ in $G$. For $P,Q \subseteq V(G)$, we define $d_G(P, Q) = \min_{u \in P,v \in Q} \{d_G(u,v)\}$.

Given a graph $G = (V, E)$, and $Y \subseteq V(G)$, we denote by $G[Y]$ the subgraph of $G$ induced by the vertex set $Y$.
A graph $G$ is called a {\em bipartite graph} if there exists a partition of $V(G) = A \uplus B$ such that for every edge $uv \in E(G)$, $u \in A$ and $v \in B$.
A graph $G$ is called a {\em split graph} if its vertex set can be partitioned into two parts $V(G) = C \uplus I$ such that $C$ is a clique and $I$ is an independent set.
A graph is called a {\em cactus graph} if every edge of the graph is contained in at most one cycle.
Let $A$ be a set of three arbitrary vertices of a graph $G$.
Then, $A$ is called an {\em asteroidal triple (AT)} if between every two vertices of $A$, there is a path avoiding the third vertex.
A {\em chord} of a cycle $C = v_0 v_1 \dotsc v_{p}v_0$ is an edge $(v_i, v_j)$ with $|i - j| \geq 2$. A chordless cycle is a cycle having no chord. A {\em hole} is a chordless cycle of length at least 4. 
A graph is called a {\em chordal graph} if it has holes as induced subgraphs.
A graph is called an {\em interval graph} if it is chordal and AT-free.
Alternatively, any interval graph has an interval representation.
It means that every vertex of an interval graph can be represented as an interval on the real line and two vertices are adjacent if and only if the intervals representing the corresponding vertices intersect.
A graph is called a {\em proper interval graph} if it is an interval graph with an interval representation such that no interval properly contains any other interval.
A graph is called a {\em bipartite permutation graph} if it is bipartite and AT-free.

A {\em sunflower} with $k$ {\em petals} and {\em core $Y$} is a family of sets $\{S_1, \dotsc , S_k\}$ such that $S_i \cap S_j = Y$ for all $i\neq  j$. The sets $S_i \setminus$ $Y$ are petals and we require none of them to be empty. We have the following lemma.

\begin{lemma}\label{lemma:sunflower-lemma}
Let $\FF$ be a family of sets over a universe $U$, such that each set in $\FF$ has cardinality exactly $d$. If $|\FF| > d!(k-1)^d$, then $\FF$ contains a sunflower with $k$ petals and such a sunflower can be computed in time polynomial in $|\FF|, |U|$ and $k$.
\end{lemma}
\paragraph{\bf Parameterized Complexity:} A parameterized problem $L$ is a subset of $\Sigma^* \times \mathbb{N}$ for some finite alphabet $\Sigma$.
An instance of a parameterized problem $L$ is a pair $(x, k) \in \Sigma^* \times \mathbb{N}$ where $k$ is called the parameter and $x$ is called the input.
We assume that $k$ is given in unary and without loss of generality $k \leq |x|$, such that $|x|$ is the length of the input.

\begin{definition}[Fixed-Parameter Tractability]
A parameterized problem $L \subseteq \Sigma^* \times \mathbb{N}$ is said to be {\em fixed-parameter tractable} if there exists an algorithm $\AAA$ that given an input $(x, k)$, runs in $f(k)|x|^{c}$ time and correctly decides if $(x, k) \in L$ where $c$ is a fixed constant independent of $|x|$ and $k$.
\end{definition}

A closely related notion to fixed-parameter tractability is the notion of Kernelization defined below.

\begin{definition}[Kernelization]
Let $L \subseteq \sum^* \times \mathbb{N}$ be a parameterized language. Kernelization is a procedure that replaces the input instance $(I,k)$ by a reduced instance $(I^\prime, k^\prime)$ such that
\begin{itemize}
\item $k^\prime \leq f(k)$, $\vert I^\prime \vert \leq g(k)$ for some function $f, g$ depending only on $k$.
\item $(I, k) \in L$ if and only if $(I^\prime, k^\prime) \in L$.
\end{itemize} 
The reduction from $(I, k)$ to $(I^\prime, k^\prime)$ must be computable in $poly(\vert I \vert + k)$ time. If $g(k) = k^{\OO(1)}$ then we say that $L$ admits a {\it polynomial kernel}.
\end{definition}

For more details on parameterized complexity, we refer to \cite{cygan2015parameterized}.

We now end the section by introducing some notations and observations on {\pionepitwodeletion} that we use in the paper. Throughout this paper, we assume that the graphs in the forbidden families $\FF_1$ and $\FF_2$ associated to {\pionepitwodeletion} are connected which is true for most of the well-known graph classes. We use  $\Pi_{(1,2)}$ to denote the class of graphs whose connected components are in the graph classes $\Pi_1$ or $\Pi_2$.

\begin{definition}[Minimal Forbidden Family]
\label{defn:minimal-forbidden-family}
 A forbidden family $\FF$ for a graph class $\Pi$ is said to be {\em minimal} if for all graphs $H \in \FF$, we have that $\FF - \{H\}$ is not a forbidden family for $\Pi$.
\end{definition}

Let $\FF_1 \times \FF_2 = \{(H_1, H_2) : H_1 \in \FF_1 \textrm{ and } H_2 \in \FF_2 \}$. The following characterization for $\Pi_{(1,2)}$ is easy to see.

\begin{observation}
\label{lemma:characterization-pairs}
A graph $G$ is in the graph class $\Pi_{(1,2)}$ if and only if no connected component $C$ of $G$ contains $H_1$ and $H_2$ as induced graphs in $C$, where $(H_1, H_2) \in \FF_1 \times \FF_2$.
\end{observation}

Let $J_1, J_2 \subseteq V(G)$ such that $G[J_i]$ is isomorphic to graphs $H_i$ for $i \in \{1,2\}$ and $(H_1, H_2) \in \FF_1 \times \FF_2$. We call the sets $J_1$ and $J_2$ as the vertex sets of the pair $(H_1, H_2)$.

In the following lemma, we show that any solution to {\pionepitwodeletion} must either hit the vertex sets of a pair in $\FF_1 \times \FF_2$ or a path connecting them.

\begin{lemma}
\label{lemma:hitting-forbidden-pairs-with-path}
Let $J_1, J_2 \subseteq V(G)$ such that $G[J_i]$ is isomorphic to graphs $H_i$ for $i \in \{1,2\}$ and $(H_1, H_2) \in \FF_1 \times \FF_2$. Let $P$ be a path between $J_1$ and $J_2$ in the graph $G$.  Then any solution for {\pionepitwodeletion} with input graph $G$ contains one of vertices in the set $J_1 \cup J_2 \cup P$.
\end{lemma}
\begin{proof}
Suppose this is not the case. Since the graphs $H_1$ and $H_2$ are connected, the graph $G'$ induced by the set $J_1 \cup J_2 \cup P$ is a connected subgraph of $G$. If a solution $X$ does not intersect $J_1 \cup J_2 \cup P$, then $G'$ occurs in a connected component $C$ of the remaining graph $G-X$. The presence of graphs $H_1$ and $H_2$ in $C$ implies that $C$ is neither in $\Pi_1$ nor in $\Pi_2$ giving a contradiction that $X$ is a solution.
\end{proof}

We use the following reduction rule for {\pionepitwodeletion} whose correctness easily follows.

\begin{reduction rule}
\label{red-rule:removal-redundant-component}
If a connected component $C$ of $G$ is in $\Pi_1$ or in $\Pi_2$, then delete $C$ from $G$.
The new instance is $(G - V(C), k)$.
\end{reduction rule}


\section{{\fpionepitwodeletion} with forbidden paths}
\label{sec:finite-set-with-paths}

In this section, we restrict the problem to the case where both the forbidden families $\FF_1$ and $\FF_2$ are finite and there exists a path $P_\alpha$ in one of the families, say $\FF_1$ where $\alpha$ is some constant. Observe that for several natural graph classes (like cluster graphs, edgeless graphs, split cluster graphs, cographs) paths above a certain length is forbidden.

We define the problem below.

\defproblem{\fpionepitwodeletionpath}{An undirected graph $G$, and an integer $k$. Furthermore, for a fixed integer $\alpha$, the path $P_\alpha \in \FF_1$ .}{Does $G$ have a set $S$ of at most $k$ vertices such that every connected component of $G - S$ is in $\Pi_1$ or in $\Pi_2$?}

Since both $\FF_1$ and $\FF_2$ are finite, the set  $\FF_1 \times \FF_2$ is also finite.

Let us look at a pair $(H_1, H_2)$ such that the distance between its vertex sets is the smallest among all the pairs in $\FF_1 \times \FF_2$. We call such a pair the {\em closest} pair. We claim below that this distance is bounded by $\alpha$.

\begin{lemma} Let $(H_1, H_2) \in \FF_1 \times \FF_2$ be a closest pair in the graph $G$ and let  $(J_1, J_2)$ be a pair of vertex subsets corresponding to the pair. Let $P$ be a shortest path between $J_1$ and $J_2$. Then $|V(P)| \leq \alpha$.
\end{lemma}
\begin{proof}
Suppose this is not the case. Then let us look at the set $J_1'$  of the last $\alpha$ vertices of $P$ which is isomorphic to $P_\alpha \in \FF_1$. Then the pair of vertex subsets $(J_1', J_2)$ corresponds to the pair $(P_\alpha, H_2)$ in the graph $G$ with the distance between them as zero. This contradicts that $(J_1, J_2)$ is the vertex subsets of the closest pair.
\end{proof}

Hence, we have the following branching rule for closest pairs where we branch on the vertex subsets plus the vertices of the path. The correctness follows from Lemma \ref{lemma:hitting-forbidden-pairs-with-path}.

\begin{branching rule}
\label{branch-rule:simple-pair-branching}
Let $(J^*, T^*)$ be the vertex subsets of a  closest pair $(H_1, H_2) \in \FF_1 \times \FF_2$. Let $P^*$ be a path corresponding to this forbidden pair. Then for each $v \in J^* \cup T^* \cup P^*$, we delete $v$ and decrease $k$ by $1$, resulting in the instance $(G-v,k -1)$.
\end{branching rule}

Using this branching rule, we have an FPT algorithm for {\fpionepitwodeletionpath}. Let $d_i$ be the size of a maximum sized finite forbidden graph in $\FF_i$ for $i \in \{1,2\}$. 
Observe that $|(J^* \cup T^*) \cap P^*| \geq 2$ as $P^*$ be a shortest path between a vertex of $J^*$ and a vertex of $P^*$.
Therefore, $|J^* \cup T^* \cup P^*| \leq |J^*| + |T^*| + |P^*| - 2$.
As $|J^*| \leq d_1, |T^*| \leq d_2$ and $|P^*| \leq \alpha$, we have that $|J^* \cup T^* \cup P^*| \leq d_1 + d_2 + \alpha - 2$.
Let $c = d_1+d_2 + \alpha - 2 $.

\begin{theorem}
{\fpionepitwodeletionpath} can be solved in $c^k poly(n)$ time.
\end{theorem}

\begin{proof}
We describe our algorithm as follows.
Let $(G, k)$ be an input instance of {\fpionepitwodeletionpath}.
We exhaustively apply Reduction Rule~\ref{red-rule:removal-redundant-component} and Branching Rule~\ref{branch-rule:simple-pair-branching} in sequence to get an instance $(G', k')$. The algorithm finds the closest pair to apply Branching Rule~\ref{branch-rule:simple-pair-branching} by going over all pairs in $(H_1, H_2) \in \FF_1 \times \FF_2$  and going over all subsets of size $|V(H_1)| + |V(H_2)|$ of the graph (which is still a polynomial in $n$) and checking the distance between them. 

Every component of $G'$ is such that it is $\FF_1$-free or $\FF_2$-free or in other words in $\Pi_1$ or $\Pi_2$. Hence if $k' \geq 0$, we return yes-instance.
Otherwise, we return no-instance.
Since the largest sized obstruction in these rules is at most $c = d_1 + d_2 + \alpha - 2 $, the bounded search tree of the algorithm has $c^{k}$ nodes bounding the running time to $c^k poly(n)$. This completes the proof.
\end{proof}

We now describe a family of graphs such that instead of focusing that each component is free of pairs in $\FF_1 \times \FF_2$, we can check whether the graph is free of graphs in this family.

Let $\FF'$ be the minimal family of graphs (as in Definition \ref{defn:minimal-forbidden-family}) such that for each member $H \in \FF'$, there exist subsets $J_1, J_2 \subseteq V(H)$ such that $H[J_i]$ is isomorphic to $H_i \in \FF_i$ for $i \in \{1,2\}$ and $d_H(J_1, J_2) \leq \alpha-1$. 

From Lemma \ref{lemma:hitting-forbidden-pairs-with-path}, it can be inferred that any graph without any members from $\FF'$ as induced subgraphs belong to the graph class $\Pi_{1,2}$. Hence, the set $\FF'$ is the forbidden family for the graph class $\Pi_{1,2}$. We now use this family to give a kernel and approximation algorithm for {\fpionepitwodeletionpath}.

The size of any member $H \in \FF'$ is bounded by $d = d_1+d_2 + \alpha - 2 $. Suppose not. Then we can identify a vertex $v \in V(H)$ which is not part of $J_1, J_2$ and a path $P$ of length  at most $\alpha-1$ between them. But then the graph  $H \setminus \{v\}$ is also in $\FF'$ contradicting that $\FF'$ is minimal. This also proves that the size of $\FF'$ is bounded by $2^{d+1 \choose 2}$ which is the bound on the number of graphs of at most $d$ vertices.

\begin{theorem}
\sloppy {\fpionepitwodeletionpath} admits a $d$-approximation algorithm, and a $\OO(k^d)$ sized kernel.
\end{theorem}
\begin{proof}

We show that an instance of {\fpionepitwodeletionpath} can be reduced to an instance of  of {\sc $d$-Hitting Set} problem defined as follows.

\defparproblem{{\sc $d$-Hitting Set}}{A family $\SSS$ of sets over a universe $U$, where each set in $\SSS$ has size at most $d$, and an integer $k$.}{$k$.}{Does there exist a subset $X \subseteq U, |X| \leq k$ such that $X$ contains at least one element from each set in $\SSS$?}

An instance $(G, k)$ of {\fpionepitwodeletionpath}, we construct an instance $(U,\SSS,k)$ {\sc $d$-Hitting Set} problem with $U= V(G)$ and  $\SSS$ being the vertex sets of all induced graphs of $G$ that is isomorphic to a graph in $\FF$. Note that every set in $\SSS$ has size atmost $d= d_1+d_2+\alpha-2$. It is easy to see that $(G,k)$ is a yes-instance of {\fpionepitwodeletionpath} if and only if $(U,\SSS,k)$ is a yes-instance of {\sc $d$-Hitting Set}.

It is folklore that the {\sc $d$-Hitting Set} problem has a $d$-approximation algorithm and a kernel with $\OO(k^d)$ size. We adapt the techniques used in this to give a $d$-approximation algorithm and a $\OO(k^d)$ sized kernel for {\fpionepitwodeletionpath}.

%

\medskip \noindent {\bf Approximation Algorithm.} Let us define the family $\mathcal{S}$ as follows. Initially $\mathcal{S} = \emptyset$. In polynomial time, we find  a subset $T \subseteq V(G)$ such that $G[T]$ is isomorphic to a member in $\FF$. We add $T$ to $\mathcal{S}$ and update $G$ to $G-T$. We repeat this step until it is no longer applicable. 

Let $S_{OPT}$ be the  minimum sized set such that in the graph $G - S_{OPT}$, every connected component is either in $\Pi_1$ or $\Pi_2$. Let $|S_{OPT}| = OPT$. Let $S$ be the set of vertices that is present in any pair of graphs in $\mathcal{S}$. From Lemma \ref{lemma:hitting-forbidden-pairs-with-path}, we can conclude that any feasible solution of $G$ must contains a vertex from each member of the family $\mathcal{S}$. Since the members of $\mathcal{S}$ are pairwise disjoint, we have that $|S_{OPT}| \geq |\mathcal{S}|$.

We have $|S|  \leq \max_{T \in \mathcal{F}}|T| \cdot |\mathcal{S}| \leq d|S_{OPT}|$. Thus we have a $d$-approximation algorithm for {\fpionepitwodeletionpath}.

\medskip \noindent {\bf Kernel.} Given an instance $(G,k)$ of {\fpionepitwodeletionpath}, we construct an equivalent instance $(U,\SSS,k)$ {\sc $d$-Hitting Set} as described earlier. For each $d' \in [d]$, we repeatedly do the following. If $|\SSS| > d'!(k+1)^{d'}$, from Lemma \ref{lemma:sunflower-lemma}, we obtain a sunflower of size $k+2$ in polynomial time. Let $Y$ be the core of sunflower and $S_1 \setminus Y, \dotsc, S_{k+2} \setminus Y$ be the $k+2$ petals, which are non-empty. We remove $S_{k+2}$ from $\SSS$.
Let $T = \bigcup_{S \in \SSS} S$. We reduce $(G,k)$ to $(G[T],k)$.

We now claim that the reduction rule is safe. The forward direction is trivial as $G[T]$ is an induced subgraph of $G$. In the reverse direction, suppose $X$ is a solution of size $k$ of the instance $(G[T],k)$. Suppose there exist a subset $U \subseteq V(G) \setminus X$ such that $G[U] \in \FF$. If $U \subseteq V(G[T])$, it contradicts that $X$ is a solution of $(G[T],k)$. In the other case, $U$ part of a sunflower. Let this sunflower be $S_1, S_2, \dotsc, S_{k+1}, U$ with core $Y$. Since graphs induced by $S_1, \dotsc, S_{k+1}$ in $G[T]$ are in $\FF$, $X$ should intersect each of them. Suppose $X$ does not intersect $Y$. Then $X$ must contain an element from each of the petals $S_1 \setminus Y, \dotsc, S_{k+1} \setminus Y$. However, since $S_i \cap S_j = Y$ for all $i,j \in [k+1], i \neq j$, we have $|X| > k$, a contradiction. Thus $X$ must contain an element from $Y$. But in this case, $X$ intersects $U$ as well, giving a contradiction.

When the reduction rule is not applicable, we have $|\SSS| \leq d!(k+1)^d$ and thus, $|T| \leq d \cdot d!(k+1)^d$. Thus we have a kernel with $\OO(k^d)$ vertices.
\end{proof}

\section{{\pionepitwodeletion} with a constant number of forbidden pairs}
\label{sec:constant-forbidden-pair}

\subsection{Forbidden Characterization for {\pionepitwodeletion}}
\label{sec:forbidden-characterization}


%

Unfortunately, the algorithm in Section \ref{sec:finite-set-with-paths} does not work when at least one of the sets $\FF_1$ or $\FF_2$ is infinite as the family $\FF_1 \times \FF_2$ is no longer finite. But we observed that for many problems, branching on most of the pairs in $\FF_1 \times \FF_2$ could be avoided.

We aim to identify such `redundant' pairs in $\FF_1 \times \FF_2$. Instead of ensuring that such pairs are absent in a graph for an instance of {\pionepitwodeletion}, we identify some graphs which are forbidden in such a graph. Since hitting the vertices of of forbidden induced graphs is the familiar framework in designing algorithms for vertex deletion problems, such a characterization for $\Pi_{(1,2)}$ with forbidden graphs and irredundant pairs would be useful.

\sloppy For example, let $\FF_1 = \{C_3, C_4 \}$ and $\FF_2 = \{D_4, C_4 \}$ where $D_4$ is the graph obtained after removing an edge from $K_4$. We have $\FF_1 \times \FF_2 = \{(C_3, D_4),(C_3, C_4),(C_4, D_4), (C_4, C_4)\}$. Since $C_4 \in \FF_1 \cap \FF_2$, the graph $C_4$ is forbidden for the graph class $\Pi_{1,2}$. Hence the pairs $(C_4, C_4), (C_3, C_4)$ and $(C_4, D_4)$ are redundant by identifying that $C_4$ is forbidden. Now note that $C_3$ is an induced subgraph of the graph $D_4$. Hence the pair $(C_3, D_4)$ can also be made redundant by identifying that $D_4$ is forbidden for $\Pi_{(1,2)}$.

We now formalize such forbidden graphs in the graph class $\Pi_{(1,2)}$ by defining the notion of super-pruned family.
Recall that if a family of graphs is {\em minimal}, no element of it is an induced subgraph of some other element of the family. 

\begin{definition}[Super-Pruned Family]

An element of a {\em super-pruned family} $\ssp(\GG_1, \GG_2)$ of two minimal families of graphs $\GG_1$ and $\GG_2$ is a graph that $(i)$ belongs to one of the two families and $(ii)$ has an element of the other family as induced subgraph. 

The family $\ssp(\GG_1, \GG_2)$ can be obtained from an enumeration of all pairs in $\GG_1 \times \GG_2$ and adding the supergraph if one of the graphs is an induced subgraph of the other. The family obtained is made minimal by removing the elements that are induced subgraphs of some other elements.
\end{definition}

For example, let ($\Pi_1$,$\Pi_2$) be (Interval, Trees), with the forbidden families $\FF_1 = \{ \textrm{net, sun, long-claw, whipping top,} \dagger\textrm{-AW}, \ddagger\textrm{-AW} \} \cup \{C_i:i \geq 4\}$ (See Figure \ref{fig:interval-obstruction}) and $\FF_2$ as the set of all cycles. Note that all graphs $C_i$ with $i \geq 4$ are in $\ssp(\FF_1 , \FF_2)$ as they occur in both $\FF_1$ and $\FF_2$. The remaining pairs of $\FF_1 \times \FF_2$ contain triangles from $\FF_2$. If the graph from $\FF_1$ is a net, sun, whipping top, $\dagger\textrm{-AW}$ or $\ddagger\textrm{-AW}$, it contains triangle as an induced subgraph. Hence these graphs are also in the family $\ssp(\FF_1 , \FF_2)$. 

We now show that graphs in $\ssp(\FF_1 , \FF_2)$  are forbidden in the graph class $\Pi_{(1,2)}$. 

\begin{lemma}
\label{lemma:super-prune-forbidden}
If a graph $G$ is in the graph class $\Pi_{(1,2)}$, then no connected component of $G$ contains a graph in $\ssp(\FF_1 , \FF_2)$ as induced subgraphs.
\end{lemma}
\begin{proof}
Suppose a graph $H \in \ssp(\FF_1 , \FF_2)$ occur as induced subgraph of a connected component $C$ of  $G$. From the definition of Super-Pruned Family, we can associate a pair 
$(H_1, H_2) \in \FF_1 \times \FF_2$ to $H$ such that either $H$ is isomorphic to $H_1$ and $H_2$ is an induced subgraph of $H_1$ or vice-versa. 
Without loss of generality, let us assume the former. Since $H_i \in \FF_i$, we know that $C$ is not in the graph class $\Pi_i$ for $i \in \{1,2\}$.
This contradicts that $G$ is in the graph class $\Pi_{(1,2)}$.
\end{proof}




Hence any pair containing a graph from $\ssp(\FF_1 , \FF_2)$ are redundant. But $\ssp(\FF_1 , \FF_2)$ does not capture all the pairs in $\FF_1 \times \FF_2$. We now define the following family to capture the remaining pairs.

\begin{definition}[Forbidden Pair Family]
\label{forbpair}
A {\em forbidden pair family} $\FF_p$, of $\FF_1$ and $\FF_2$, consists of all pairs $(H_1,H_2) \in \FF_1 \times \FF_2$ such that both $H_1 \notin \ssp(\FF_1 , \FF_2)$ and $H_2 \notin \ssp(\FF_1 , \FF_2)$.
\end{definition}

Informally, a pair $(H_1, H_2)$ is part of a forbidden pair family $\FF_p$ if at least one of $H_1$ and $H_2$ does not belong to $\ssp(\FF_1, \FF_2)$.
For example, if $\Pi_1$ is the class of interval graphs and $\Pi_2$ is the class of forests, we have already shown that $\ssp(\FF_1 , \FF_2)$ contains all the graphs in $\FF_1$ except long-claw. The only remaining pair is $($long-claw, triangle$)$ and the singleton set containing this pair forms the forbidden pair family.
Now we characterize $\Pi_{(1,2)}$ based on the super-pruned family and the forbidden pair family associated with $\FF_1$ and $\FF_2$ as follows.
This is used in the algorithms in Section \ref{sec:constant-forbidden-pair}.

\begin{lemma}
\label{lemma:forbidden-pair-characterization}
The following statements are equivalent.
\begin{itemize}
\item Each connected component of $G$ is either in $\Pi_1$ or $\Pi_2$.
\item  The graph $G$ does not contain graphs in the super-pruned family $\ssp(\FF_1 , \FF_2)$ as induced subgraphs. Furthermore, for pairs $(H_1, H_2)$ in the forbidden pair family of $\FF_1$ and $\FF_2$,  $H_1$ and $H_2$ both cannot appear as induced subgraphs in a connected component of $G$.
\end{itemize}
\end{lemma}

\begin{proof}
To  prove the forward direction, note that from Lemma \ref{lemma:super-prune-forbidden}, the $G$ does not contain graphs in the super-pruned family $\ssp(\FF_1 , \FF_2)$ as induced subgraphs.
Hence, suppose that there exists a pair $(H_1, H_2) \in \FF_p$ in a connected component $\chi$ of $G$. But then $\chi$ cannot be in $\Pi_1$ due to the presence of $H_1$ and cannot be in $\Pi_2$ due to the presence of $\Pi_2$ giving a contradiction.
To prove the converse, suppose that $G$ contains a component $\chi$ which is neither in $\Pi_1$ nor in $\Pi_2$. Then there exist graphs $H_1 \in \FF_1$ and $H_2 \in \FF_2$ occuring as induced subgraphs of $\chi$. If $H_1$ occurs as an induced subgraph of $H_2$ or vice-versa, then the supergraph occurs in $\ssp(\FF_1 , \FF_2)$ giving a contradiction. Else we have $H_1 \in \FF_1 \setminus \ssp(\FF_1 , \FF_2)$ and $H_2 \in \FF_2 \setminus \ssp(\FF_1 , \FF_2)$. Hence $(H_1, H_2) \in \FF_p$ giving a contradiction.
\end{proof}

We now define useful notions of forbidden sets and closest forbidden pairs for the graph class $\Pi_{(1,2)}$.

\begin{definition} We call a minimal vertex subset $Q \subseteq V(G)$ as a {\em forbidden set} corresponding to the graph class $\Pi_{(1,2)}$ if $G[Q]$ is isomorphic to a graph in  $\ssp(\FF_1 , \FF_2)$ or $G[Q]$ is connected and contains both $H_1$ and $H_2$ as induced subgraphs for some forbidden pair $(H_1, H_2)$ of $\Pi_{(1,2)}$.
\end{definition}

\begin{definition}
We say that a forbidden pair $(H_1, H_2)$ is a {\em closest forbidden pair} in a graph $G$ if there exists subsets $J_1, J_2 \subseteq V(G)$ such that $G[J_1]$ is isomorphic to $H_1$, $G[J_2]$ is isomorphic to $H_2$ and the distance between $J_1$ and $J_2$ in $G$ is the smallest among all such pairs 
over all forbidden pairs of $\FF_1$ and $\FF_2$. We call the pair of vertex subsets $(J_1, J_2)$ as the {\em vertex subsets corresponding to the closest forbidden pair}. We call a shortest path $P$ between $J_1$ and $J_2$ as the {\em path corresponding to the closest forbidden pair}.
\end{definition}


 
\subsection{The case with forbidden paths}
\label{subsection-forbidden-paths}

We now aim to give an algorithm for {\pionepitwodeletion} when the forbidden families $\FF_1$ and $\FF_2$ be infinite but the forbidden pair family $\FF_p$ is finite. We also assume that $P_\alpha \in \FF_1$ and $\FF_p$ is given as input.

Let us list the conditions that {\pionepitwodeletion} is required to satisfy.

\begin{enumerate}
\item\label{1problem-condition-1} The vertex deletion problems for the graph classes $\Pi_1$ and $\Pi_2$ are FPT with algorithms to the respective classes being $\mathcal{A}_1$ and $\mathcal{A}_2$. 
\item\label{1problem-condition-2} $\mathcal{F}_p$, the forbidden pair family of $\FF_1$ and $\FF_2$ is of constant size.
\item\label{1problem-condition-4} The path $P_\alpha \in \FF_1$.
\end{enumerate}

\defproblem{\PalphafreePiOnePiTwo}{An undirected graph $G$, graph classes $\Pi_1, \Pi_2$ with associated forbidden families $\FF_1$ and $\FF_2$ such that Conditions \ref{1problem-condition-1} - \ref{1problem-condition-4} are satisfied and an integer $k$.}{Does $G$ have a set $S$ of at most $k$ vertices such that every connected component of $G - S$ is either in $\Pi_1$ or in $\Pi_2$?}

Since the forbidden pair set is finite, we have the following Branching Rule for closest forbidden pairs which is similar to Branching Rule \ref{branch-rule:simple-pair-branching} where we branch on the vertex subsets plus the vertices of the path. The correctness follows from Lemma \ref{lemma:hitting-forbidden-pairs-with-path}.

\begin{branching rule}
\label{branch-rule:closest-pair-branching-path}
Let $(J^*, T^*)$ be the vertex subsets of a  closest forbidden pair $(H_1, H_2) \in \FF_1 \times \FF_2$. Let $P^*$ be a path corresponding to this forbidden pair. Then for each $v \in J^* \cup T^* \cup P^*$, we delete $v$ and decrease $k$ by $1$, resulting in the instance $(G-v,k -1)$.
\end{branching rule}

From here on, assume that $(G, k)$ be an instance at which Reduction Rule~\ref{red-rule:removal-redundant-component} and Branching Rule~\ref{branch-rule:closest-pair-branching-path} are not applicable. Note that any component of $G$ is now free of forbidden pairs.

 Let $\mathcal{F}_p^1$ denote the family of graphs $H_1$ where $(H_1, H_2) \in \mathcal{F}_p$. Similarly define $\mathcal{F}_p^2$ as the family of graphs $H_2$ where $(H_1, H_2) \in \mathcal{F}_p$.  By the definition of forbidden pairs, the set of pairs $(H_1, H_2)$ with $H_i \in \mathcal{F}_p^i$ is the forbidden pair set $\FF_p$. Hence a graph that does not contain any forbidden pairs  is $\mathcal{F}_p^1$-free or $\mathcal{F}_p^2$-free. With this observation, the following results are easy to see.

\begin{lemma}
\label{lemma:path-component-with-no-H1}
Let $C$ be a connected component of $G$ that is $\mathcal{F}_p^i$-free for some $i \in \{1,2\}$. If $G[C]$ has no $\Pi_i$ vertex deletion set of size $k$, then $(G,k)$ is a no-instance. Otherwise, let $X$ be a minimum $\Pi_i$ vertex deletion set of $G[C]$. Then $(G, k)$ is a yes-instance  if and only if $(G - V(C), k - |X|)$ is a yes-instance.
\end{lemma}
\begin{proof}
Suppose that the premise of the statement holds and $k' = k - |X|$.

$(\Leftarrow)$ The backward direction is trivial.
If $G - V(C)$ has a feasible solution $S'$ of size at most $k'$, then we can add the minimum sized $\Pi_i$ vertex deletion set $X$ of $G$ and output $S' \cup X$ has a feasible solution of size $k' + |X| = k$.

$(\Rightarrow)$ We prove the forward direction now.
Suppose that $S^*$ be a feasible solution of size at most $k$ to $(G, k)$ and let $Y = S^* \cap V(C)$.
We prove that $D = (S^* \setminus Y) \cup X$ is also a feasible solution to $(G, k)$ and $|Y| \geq |X|$.
If we manage to prove that $Y$ is a $\Pi_i$-deletion vertex set of $C$ then we are done. This is because since $X$ is a minimum $\Pi_i$-deletion vertex set of $C$, $|X| \leq |Y|$. 

We now prove that $Y$ is indeed a $\Pi_i$-deletion vertex set of $C$. Suppose not. Then there exist a forbidden set $Q$ in $C-Y$. Note that $C$ does not contain any forbidden pairs. Hence from Lemma \ref{lemma:forbidden-pair-characterization}, $Q$ is isomorphic to a graph in $\ssp(\FF_1 , \FF_2)$. But in that case, from the definition of Super-Pruned Family, any graph  $H \in \ssp(\FF_1 , \FF_2)$ contains an induced subgraph which is isomorphic to some graph in $\FF_i$. Hence the presence of $Q$ contradicts that $Y$ is a $\Pi_i$-deletion set. This completes the proof.
\end{proof}

We are ready to prove our main theorem statement of this section. Let $f(k) = \max \{f_1(k), f_2(k)\}$ where $f_i(k) poly(n)$ is the running time for the algorithm $\mathcal{A}_i$. Also let $c$ be the maximum among the size of graphs in $\mathcal{G}_1$ and the integer $\max_{(H_1, H_2) \in \mathcal{F}_p}(|H_1|+|H_2|+ \alpha -2)$. 

\begin{theorem}
\label{theorem:finite-path-FPT-algorithm}
{\PalphafreePiOnePiTwo} can be solved in \\ ${\max\{f(k),c^k\}) poly(n)}$-time.
\end{theorem}

\begin{proof}
We describe our algorithm as follows.
Let $(G, k)$ be an input instance of {\PalphafreePiOnePiTwo}.
We exhaustively apply Reduction Rule~\ref{red-rule:removal-redundant-component} and Branching Rule~\ref{branch-rule:closest-pair-branching-path} in sequence to get an instance $(G', k')$. The algorithm finds the closest pair to apply Branching Rule~\ref{branch-rule:closest-pair-branching-path} by going over all pairs in $(H_1, H_2) \in \FF_p$  and going over all subsets of size at most $|V(H_1)| + |V(H_2)|$ of the graph (which is still a polynomial in $n$) and checking the distance between them. 
Since the largest sized obstruction in these rules is at most $c$, the bounded search tree of the algorithm so far has $c^{k - k'}$ nodes.
Hence, every component of $G'$ is such that it is $\mathcal{F}_p^1$-free but has graphs in $\mathcal{F}_p^2$ as induced subgraphs, or vice-versa.
In the first case, we invoke the $f_1(|X|)poly(n)$-time algorithm for {\sc $\Pi_1$ Vertex Deletion} on $G[C]$ to compute a minimum $\Pi_1$ vertex deletion set $X$ of $G[C]$.
In the second case, we invoke the $f_2(|X|) poly(n)$-time algorithm for {\sc $\Pi_2$ Vertex Deletion}  to compute a minimum $\Pi_2$ vertex deletion set $X$ of $G[C]$. The correctness follows from Lemma \ref{lemma:path-component-with-no-H1}. 
This creates the total number of nodes in the search tree to $c^{k-k'}f(k')$, bounding the running time to $c^{k-k'}f(k')poly(n)$.
This completes the proof.
\end{proof}

We now give an approximation algorithm for {\PalphafreePiOnePiTwo} when for $i \in \{1,2\}$, {\sc $\Pi_i$ Vertex Deletion} has  an approximation algorithm with approximation factor $c_i$. 

\begin{theorem}
\label{theorem:finite-path-approximation-algorithm}
{\PalphafreePiOnePiTwo} has a $d$-approximation algorithm where $d=\max \{c, c_1, c_2\}$.
\end{theorem}

\begin{proof}
Let $G$ be the input graph. Let $S_{OPT}$ be the  minimum sized set such that in the graph $G - S_{OPT}$, every connected component is either in $\Pi_1$ or $\Pi_2$. Let $|S_{OPT}| = OPT$.


Let us define the family $\mathcal{S}_1$ as follows. Initially $\mathcal{S}_1 = \emptyset$. In polynomial time, we find the closest forbidden pair $(J^*,T^*)$ in $G$ with $P^*$ being a shortest path between the pair, add $J^* \cup T^* \cup P^*$ to $\mathcal{S}_1$ and delete $J^* \cup T^* \cup P^*$ from $G$. We repeat this step until it is no longer applicable. Let $S_1$ be the set of vertices that is present in any pair of graphs in $\mathcal{S}_1$. From Lemma \ref{lemma:hitting-forbidden-pairs-with-path}, we can conclude that any feasible solution of $G$ must contains a vertex from each member of the family $\mathcal{S}_1$. Since the members of $\mathcal{S}_1$ are pairwise disjoint, we have that $|S_{OPT} \cap S_1| \geq |\mathcal{S}_1|$.

Let $G' = G - S_1$. We now construct a set $S_2$ as follows. Let $C_1, \dotsc , C_q$ be the connected components of $G'$. If a connected component $C_i$ has no graphs in $\mathcal{F}_p^j$ as induced subgraph for $j \in \{1,2\}$, we apply the $c_j$-approximation algorithm for {\sc $\Pi_j$ Vertex Deletion} on $G'[C_i]$ to obtain a solution $Z_i$. The correctness comes from Lemma \ref{lemma:path-component-with-no-H1}. We have $S_2 = \bigcup_{i \in [q]} Z_i$. Since $(S_{OPT} - S_1) \cap C_i$ is a feasible solution for the connected component $C_i$ of $G'$, we have that $|Z_i| \leq (\max \{c_1,c_2\})|(S_{OPT} - S_1) \cap C_i|$ for all $i \in [q]$.

We set $S = S_1 \cup S_2$. We have 
\begin{eqnarray*}
|S|  &=&  |S_1| + |S_2| \\
 & \leq & \left(\max_{(H_1, H_2) \in \mathcal{F}_p}(|H_1|+|H_2|+ \alpha -2 )\right)|\mathcal{S}_1| +  \sum\limits_{i=1}^{q} |Z_i|\\
& \leq & c|S_{OPT} \cap S_1| +  \\ & &\sum\limits_{i=1}^{q} (\max \{c_1,c_2\})|(S_{OPT} - S_1) \cap C_i|\\
& \leq & (\max \{c, c_1, c_2\})|S_{OPT}|
\end{eqnarray*}

Thus we have a $d$-approximation algorithm for {\PalphafreePiOnePiTwo}.
\end{proof}

We now give some examples of {\PalphafreePiOnePiTwo}.

\subsubsection{Cliques or $K_t$-free graph subclass}

We focus on the case of  {\pionepitwodeletion} when $\Pi_1$ is the class of cluster graphs(where every connected component of the graph is a clique) and $\Pi_2$ is any graph class such that the complete graph $K_t$ for some constant $t$ is forbidden in this graph and the problem {\sc $\Pi_2$ Vertex Deletion} is known to be FPT. We show that this problem is an example of {\PalphafreePiOnePiTwo}. We will see later that $\Pi_2$ can be many of the popular classes including planar graphs, cactus graphs, $t$-treewidth graph.

Let us formalize the problem.

 \defparproblem{{\CliqueorKtFree}}{An undirected graph $G = (V, E)$, an integer $k$ and $\Pi_2$ is the graph class which is $K_t$-free for some constant $t$ and {\sc $\Pi_2$ Vertex Deletion} has an FPT algorithm $\mathcal{A}$ with running time $f_2(k')poly(n)$ for solution size $k'$.}{$k$}{Is there $S \subseteq V(G)$ of size at most $k$ such that every connected component of $G - S$ is either a clique graph, or in $\Pi_2$?}
 
We now show that Conditions \ref{1problem-condition-1} - \ref{1problem-condition-4}  are satisfied by {\CliqueorKtFree}. The {\sc Clique Vertex Deletion} problem is same as {\sc Vertex Cover} in the complement graph; thus has  a $1.27^k poly(n)$ \cite{chen2010improved} time algorithm parameterized by the solution size $k$. The {\sc $\Pi_2$ Vertex Deletion} has an $f_2(k)poly(n)$ time FPT algorithms parameterized by the solution size $k$. Hence Condition \ref{1problem-condition-1} is satisfied by  {\CliqueorKtFree}. Since $P_3 \in \FF_1$, Condition \ref{1problem-condition-4} is satisfied by  {\CliqueorKtFree}. The only condition remaining to be proven is Condition \ref{1problem-condition-2} which we do below.

\begin{lemma} 
\label{lemma:CliqueorKtFree-condition-2}
Condition \ref{1problem-condition-2} is satisfied by {\CliqueorKtFree}.
\end{lemma}
\begin{proof}
Let us first infer what the forbidden pair family corresponding to {\CliqueorKtFree} is. We have the forbidden family for $\Pi_1$ as $\FF_1 = \{P_3\}$. Note that we don't know what the forbidden family $\FF_2$ for the graph class $\Pi_2$ is. We only know that the graph $K_t$ is present in $\FF_2$. A crucial observation is that this is all needed to infer that the forbidden pair family for {\CliqueorKtFree}.

We know that a graph is $P_3$-free if and only if it is a collection of cliques. Hence, every graph $H \in \FF_2$ that is not a collection of cliques, contains $P_3 \in \FF_1$ as induced graphs. Hence all such graphs $H$ belong to the Super-Pruned family $\ssp(\FF_1 , \FF_2)$. We also know that every graph in $\FF_2$ is connected. Hence, the only graphs possibly in $\FF_2$ that is $P_3$-free are the clique graphs $K_r$ with $r \neq t$. Note that $K_r$ with $r>t$ contains $K_t$ as an induced subgraph. Thus, we can ignore these cliques as well. 
Hence the forbidden pair family is the set of pairs $(P_3, K_r), 1 \leq r \leq t$ which is of size at most $t$.
\end{proof}

Since all the conditions of {\PalphafreePiOnePiTwo} is satisfied by {\CliqueorKtFree}, we have the following theorem.

\begin{theorem}
{\CliqueorKtFree} has an FPT algorithm with running time $\max \{(t+2)^k, f_2(k)\}poly(n)$.
\end{theorem}
%

We also give an approximation algorithm for {\CliqueorKtFree}. But we appropriately change the definition of  {\CliqueorKtFree} where instead of  the assumption that {\sc $\Pi_2$ Vertex Deletion} has an FPT algorithm with running time $f_2(k')poly(n)$ for solution size $k'$, we assume that {\sc $\Pi_2$ Vertex Deletion} has a polynomial time approximation algorithm with approximation factor $f_2$.

Observing that  $d=\max \{c, c_1, c_2\} =\max \{t+2, f_2\}$ from Theorem \ref{theorem:finite-path-approximation-algorithm}, we have the following theorem.

\begin{theorem}
{\CliqueorKtFree} has an approximation algorithm with approximation factor $\max \{t+2, f_2\}$.
\end{theorem}

We now give examples for the graph class $\Pi_2$ in {\CliqueorKtFree} resulting in FPT and approximation algorithms for the corresponding problems.

\begin{itemize}
\item Let $\Pi_2$ be the class of trees. Since triangles are forbidden in trees, we have $t=3$. The problem {\sc $\Pi_2$ Vertex Deletion} corresponds to {\sc Feedback Vertex Set} which has a $3.619^k poly(n)$ time FPT algorithm \cite{KP14} and a $2$-approximation algorithm \cite{bafna19992}. We notice here that the closest $(P_3, C_3)$ pair always intersect on at least two vertices. Hence the branching factor of Branching Rule \ref{branch-rule:closest-pair-branching-path} can be improved to $t+1 = 4$ in this case.

\item Let $\Pi_2$ be the class of cactus graphs.
 The graph $K_4$ is forbidden in cactus graphs as it has two triangles sharing an edge. Hence  we have $t=4$. The problem {\sc $\Pi_2$ Vertex Deletion} corresponds to {\sc Cactus Vertex Deletion} which has a $17.64^k poly(n)$ time FPT algorithm \cite{aoike2022improved}.
\item Let $\Pi_2$ be the class of planar graphs. Since $K_5$ is not planar, we have $t=5$. The problem {\sc $\Pi_2$ Vertex Deletion} corresponds to {\sc Planar Vertex Deletion} which has an $k^{O(k)}poly(n)$ time  FPT algorithm \cite{jansen2014near} and a $log^{O(1)}n$-approximation algorithm with running time $n^{O(\log n/ \log log n)}$ \cite{kawarabayashi2017polylogarithmic}.
\item Let $\Pi_2$ be the class of $\eta$-treewidth graphs. Since $K_{\eta+2}$ has treewidth $\eta+1$, it is forbidden in such graphs. The problem {\sc $\Pi_2$ Vertex Deletion} corresponds to {\sc $\eta$ Treewidth Vertex Deletion} which has a $2^{O(k))}n^{\OO(1)}$ time FPT algorithm and a $\OO(1)$-approximation algorithm \cite{fomin2012planar}.

We have the following corollary.

\begin{corollary} {\pionepitwodeletion} when $\Pi_1$ is the class of cliques and $\Pi_2$ is the class of

\begin{itemize}

\item trees has an $4^{k}poly(n)$ time  FPT algorithm and a $4$-approximation algorithm.

\item cactus graphs has an $17.64^{k}poly(n)$ time  FPT algorithm.
\item planar graphs has an $k^{O(k)}poly(n)$ time  FPT algorithm and a $log^{O(1)}n$-approximation algorithm with running time $n^{O(\log n/ \log log n)}$.
\item $\eta$-treewidth graphs has a $\max{ \{ (\eta+4)^k, 2^{O(k))} \} }n^{\OO(1)}$ time FPT algorithm and a $\OO(1)$-approximation algorithm.

\end{itemize}
\end{corollary}
\end{itemize}

\subsubsection{Split or Bipartite Graphs}

\defparproblem{{\SplitorBipartite}}{An undirected graph $G = (V, E)$, an integer $k$.}{$k$}{Is there $S \subseteq V(G)$ of size at most $k$ such that every connected component of $G - S$ is either a split graph or bipartite?}

 The family $\FF_1$ for graphs whose each connected component is a split graph is $\FF_1 = \{C_4, C_5, P_5,$ necktie, bowtie$\}$  \cite{bruckner2015graph}. A necktie is the graph with vertices $\{a,b,c,d,e\}$ where $\{a,b,e\}$ forms a triangle and $\{a,b,c,d\}$ forms a $P_4$. A bowtie is the graph obtained from a necktie by adding the edge $(b,d)$. The family $\FF_2$ for graphs whose each connected component is a bipartite graph is the set of odd cycles.

The problems {\sc Split Vertex Deletion} and {\sc Odd Cycle Transversal} has  $1.27^k k^{\OO(\log k)}poly(n)$ \cite{cygan2013split} and $2.314^k poly(n)$ time \cite{lokshtanov2014faster} FPT algorithms parameterized by the solution size $k$ respectively. Hence Condition \ref{1problem-condition-1} is satisfied by  {\SplitorBipartite}. Since $P_5 \in \FF_1$, Condition \ref{1problem-condition-4} is satisfied by  {\SplitorBipartite}.

The graph $C_5$ is common in both families whereas $P_5$ is an induced subgraph of odd cycles $C_i$ with $i \geq 7$. Hence  the family $\ssp( \FF_1, \FF_2)$ contains all members in $\FF_2$ except $C_3$. Since both necktie and bowtie contains triangles, they are part of $\ssp( \FF_1, \FF_2)$ as well. The only remaining pairs are $(C_4, C_3)$ and $(P_5, C_3)$ which forms the forbidden pair family $\FF_p$. Since it is of size two, Condition \ref{1problem-condition-2} is satisfied.

Hence, we have the following corollary from the algorithms for {\PalphafreePiOnePiTwo}. Observing that the largest obstruction set that we branch on is for the pair $(C_4, C_3)$ with a path of length at most four between them, we have $c=11$. The {\sc Split Vertex Deletion} problem has a $5$-appoximation algorithm as it can be written as an instance of {\sc $5$-Hitting Set}. The problem {\sc Odd Cycle Transversal} has a $\log(OPT)$ approximation algorithm \cite{garg1996approximate} where $OPT$ is the size of the optimal solution. The latter dominates the approximation factor $d$ of the algorithm obtained from Theorem \ref{theorem:finite-path-approximation-algorithm}.

\begin{corollary}
{\SplitorBipartite} can be solved in \\ ${11^k poly(n)}$-time and has a $\log(OPT)$ approximation algorithm where $OPT$ is the size of the optimal solution.
\end{corollary}

\subsection{Algorithms for {\pionepitwodeletion} without forbidden paths }
\label{subsection:constant-forbidden-pair}

We have seen examples of {\pionepitwodeletion} where even though the families $\FF_i$ are infinite, we manage to come up with fast FPT algorithms. This is mainly thanks to Branching Rule \ref{branch-rule:closest-pair-branching-path} whose branching factor is bounded due to the fact that the path between them is bounded. We now look at examples of {\pionepitwodeletion} where paths are not present in the sets $\FF_i$. Hence the path between the closest forbidden pair is no longer bounded. We observe that for certain pairs of graph classes, there is always an optimal solution that does not intersect the path.

We give a general algorithm for pairs of graph classes where we enforce this condition. We first look at the simple case of {\ClawfreeorTriangleFree} as a precursor to the algorithm.


\subsubsection{Claw-free or Triangle-Free graphs}
\label{section:clae-free-triangle-free}
We define the problem. 

\defparproblem{{\ClawfreeorTriangleFree}}{An undirected graph $G = (V, E)$ and an integer $k$.}{$k$}{Is there $S \subseteq V(G)$ of size at most $k$ such that every connected component of $G - S$ is either a claw-free graph, or a triangle-free graph?}

The forbidden pair family corresponding to the graph class is of size one which is $\{(K_{1,3}, C_3)\}$. We now describe a branching rule corresponding to the closest forbidden pair in the graph.

\begin{branching rule}
\label{branch-rule:claw-triangle-branching}
Let $(J^*, T^*)$ be the vertex subsets of a closest claw-triangle pair in a connected component of $G$, where $G[J^*]$ is isomorphic to a claw, and $G[T^*]$ is isomorphic to a triangle. Then for each $v \in J^* \cup T^*$, delete $v$ and decrease $k$ by $1$, resulting in the instance $(G-v,k -1)$.
\end{branching rule}

We now prove that Branching Rule \ref{branch-rule:claw-triangle-branching} is safe. Let $P^* := x_0,x_1,\ldots,x_{d-1},x_d$ be a shortest path between $J^*$ and $T^*$ of length $d_G(J^*, T^*) = d$ with $x_0 = u \in J^*$ and $x_d = v \in T^*$.
 Let $C$ be the connected component of the graph $G - (J^* \cup T^*)$ containing the internal vertices of $P^*$.  We have the following lemma.

\begin{lemma}
\label{lemma:nearest-claw-triangle-pair-path}
The graph corresponding to $C$ is the path $P^*$ without its end vertices. Furthermore, the only vertices of $C$ adjacent to $J^* \cup T^*$ are $x_1$ and $x_{d-1}$ which are only adjacent to vertices $u$ and $v$ respectively.
\end{lemma}

\begin{proof}

\begin{figure}[t]
\centering
	\includegraphics[scale=0.28]{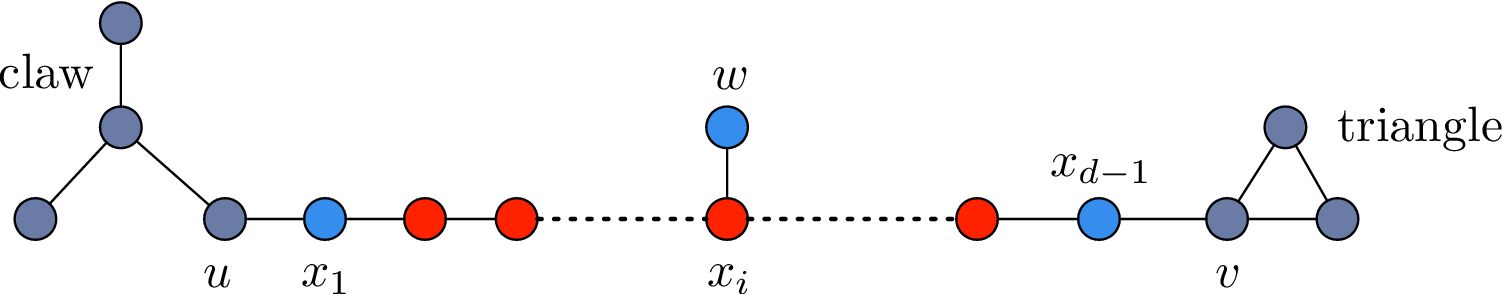}
	\caption{An illustration of a shortest path between a closest (claw, triangle) pair.}
\label{fig:claw-triangle-pair}
\end{figure}

Suppose that there is a vertex $w \in V(C) \setminus V(P^*)$ that is adjacent to a vertex $x_i \in V(P^*)$ with $1 \leq i \leq d-1$. Let us look at the graph induced by the set of vertices $\{w,x_{i-1},x_i,x_{i+1}\}$ in $G$. If $w$ is adjacent to $x_{i-1}$, then the graph induced by the vertices $T' = \{w,x_{i-1},x_i\}$ forms a triangle. Then since $d_G(J^*, T') < d_G(J^*,T^*)$, the pair $(J^*, T')$ contradicts that $(J^*, T^*)$ is the closest forbidden pair in the graph $G$. We can similarly prove that $w$ is not adjacent to $x_{i+1}$. When $w$ is not adjacent to both the vertices $x_{i-1}$ and $x_{i+1}$, the graph induced by the vertices $J'=\{w,x_{i-1},x_i,x_{i+1}\}$ forms a claw. Then since $d_G(J', T^*) < d_G(J^*,T^*)$, the pair $(J', T^*)$ contradicts that $(J^*, T^*)$ is the closest forbidden pair in the graph $G$. 

Hence we conclude that there is no vertex $w \in V(C) \setminus V(P^*)$ that is adjacent to a vertex $x_i \in V(P^*)$ with $1 \leq i \leq d-1$. Since $C$ is defined as the connected component of graph $G - (J^* \cup T^*)$ containing the internal vertices of $P^*$, the first statement of the claim follows.

The second statement is contradicted only when there is an edge from  a vertex $x_i \in V(P^*)$ with $1 \leq i \leq d-1$ to the set $J^* \cup T^*$ apart from the edges $(u,x_1)$ and $(x_{d-1},v)$ . If $2 \leq i \leq d-2$, this creates a shorter path $P'$ between $J^*$ and $T^*$ via this edge giving a contradiction. For $i=1$, in the case $x_1$ is adjacent to some vertex $w \in J^* \setminus \{u\}$, we can show that either the graph induced by the set of vertices $\{w, u, x_1, x_2\}$ is a claw or the graph induced by the set of vertices $\{w, u, x_1\}$ is a triangle, both contradicting that $(J^*, T^*)$ is the closest forbidden pair in the graph $G$.  For $i=d-1$, in the case $x_{d-1}$ is adjacent to $w \in T^* \setminus \{v\}$, we can show that the graph induced by the vertices $\{w, x_{d-1}, v\}$ is a triangle contradicting that $(J^*, T^*)$ is the closest forbidden pair in the graph $G$. This covers all the cases.
\end{proof}

\begin{lemma}
\label{lemma:branching-rule-claw-triangle-safeness}
 Branching Rule \ref{branch-rule:claw-triangle-branching} is sound.
\end{lemma}

\begin{proof}
Suppose not. In this case, all the optimal solutions for a {\ClawfreeorTriangleFree} instance is such that it does not intersect $J^* \cup T^*$.  Let $P^* := u,x_1,\ldots,x_{d-1},v$ be a shortest path between $J^*$ and $T^*$ of length $d_G(J^*, T^*) = d$. Since the graphs $G[J^*]$ and $G[T^*]$ cannot be in the same connected component after deleting the solution, any optimal solution $X$ has to intersect the set of internal vertices of $P^*$. 

We now claim that $X' = (X \setminus (P^* \setminus \{u\})) \cup \{v\}$ is also an optimal solution for $G$. Suppose not. Then there is a forbidden set $Q$ such that $Q \cap X' = \emptyset$. Without loss of generality, assume $Q$ is inclusion-wise minimal; i.e. for all $q \in Q$, the graph induced by $Q \setminus \{q\}$ is not a forbidden set. 
Note that $Q$ is a connected set and intersects some vertex $x_i \in P^* \setminus \{u,v\}$ with $1 \leq i \leq d-1$ as $X$ is a feasible solution. The graph $G[Q]$ contain a forbidden pair for $(K_{1,3},C_3)$ and hence contains a cycle. Hence $Q$ is not fully contained in the connected component $C$ of the graph $G \setminus (J^* \cup T^*)$ which is is $P^* \setminus \{u,v\}$ from Lemma \ref{lemma:nearest-claw-triangle-pair-path}. Hence we conclude that $Q$ contains some vertex outside $P^* \setminus \{u,v\}$ as well.

From Lemma \ref{lemma:nearest-claw-triangle-pair-path}, the only neighbors of $C$ is $u$ and $v$ via $x_1$ and $x_{d-1}$ respectively. Since $v \in X'$ and $Q$ is connected, we can conclude that $Q$ contains the subpath $P'$ from $u$ to $x_i$. We now claim that even after deleting the vertices of this subpath from $Q$  except $u$, the set remains forbidden. This contradicts that $Q$ is a forbidden set as it is not a minimal set.

Since $Q$ is a forbidden set, it contains vertex subsets that are isomorphic to a claw and a triangle. From Lemma \ref{lemma:nearest-long-claw-triangle-pair-path-neighbors}, we can conclude that none of the vertices of the subpath $P'$ can be part of any triangle in $G$. None of these vertices can be part of a claw in $G$ either as it contradicts that $(J^*, T^*)$ is the closest forbidden pair. Since $v \in X'$ disconnects the path $P^*$, the subpath $P'$ is not part of a path connecting a claw and a triangle either. Hence the set after removing the vertices of $P'$ from $Q$ is still a forbidden set contradicting that $Q$ is minimal.
\end{proof}

We are ready to give the algorithm for {\ClawfreeorTriangleFree}.

\begin{theorem}
\label{theorem:claw-free-triangle-free-FPT-algorithm}
{\ClawfreeorTriangleFree} can be solved in $7^k poly(n)$ time.
\end{theorem}

\begin{proof}
Let $(G, k)$ be an input instance of {\ClawfreeorTriangleFree}.
We exhaustively apply Reduction Rule~\ref{red-rule:removal-redundant-component} and Branching Rule~\ref{branch-rule:claw-triangle-branching} in sequence to get an instance $(G', k')$ such that any component of $G'$ is either claw-free or triangle-free. Note that finding the closest claw-triangle pair can be done by going over all subsets of size at most $7$, checking if they do induce a claw and a triangle and finding the shortest path between them.
The correctness follows from Lemma~\ref{lemma:branching-rule-claw-triangle-safeness}. If $k'<0$, we return no-instance. Else, we return yes-instance.


Since we branch on a set of size at most $7$ in Branching Rule~\ref{branch-rule:claw-triangle-branching}, the bounded search tree of the algorithm has at most $7^{k}$ nodes. This bounds the running time to $7^{k}n^{\OO(1)}$.
This completes the proof.
\end{proof}

We also give an approximation algorithm for {\ClawfreeorTriangleFree} using similar ideas. 

\begin{theorem}
\label{theorem:claw-free-triangle-free-approx-algorithm}
{\ClawfreeorTriangleFree}  has a $7$-approximation algorithm.
\end{theorem}

\begin{proof}
Let $G$ be the input graph. The approximation algorithm for {\ClawfreeorTriangleFree} is as follows. Let  $\mathcal{S}_1$ be a family of sets initialized to $\emptyset$. We find a closest forbidden pair  $(J^*,T^*)$ in $G'$, add $J^* \cup T^*$ to $\mathcal{S}_1$ and delete $J^* \cup T^*$ from $G'$. We repeat this step until it is no longer applicable. 
Let $S_{OPT}$ be the minimum sized set such that in the graph $G - S_{OPT}$, every connected component is either a claw-free graph or a triangle-free graph. Let $|S_{OPT}| = OPT$. Let $S_1$ be the set of vertices that is present in any set in $\mathcal{S}_1$. From the safeness proof of Branching Rule \ref{branch-rule:claw-triangle-branching}, we can conclude that any feasible solution of $G$ must contain a vertex from each set of the family $\mathcal{S}_1$. 
Since the union of all the sets in $\mathcal{S}_1$ is a feasible solution, we have that $|S_1| = 7|\mathcal{S}_1| \leq 7|S_{OPT}|$.

Thus we have a 7-approximation algorithm for \ClawfreeorTriangleFree.
\end{proof}

\subsubsection{Algorithm for {\InfCMSOPiOnePiTwo}}
\label{subsection:general-algorithm}

Now, we show that the algorithm ideas from {\ClawfreeorTriangleFree} are applicable for a larger number of pairs of graph classes. 
Later in Section \ref{sec:other-examples}, we give other examples for pairs of graph classes where the same ideas work.

We define a variant of {\pionepitwodeletion} called {\InfCMSOPiOnePiTwo} satisfying the following properties.

\begin{enumerate}
\item  \label{2problem-condition-1} The vertex deletion problems for the graph classes $\Pi_1$ and $\Pi_2$ are FPT with algorithms to the respective classes being $\mathcal{A}_1$ and $\mathcal{A}_2$. 
\item  \label{2problem-condition-2} $\mathcal{F}_p$, the forbidden pair family of $\FF_1$ and $\FF_2$ is of constant size.
\item  \label{2problem-condition-4} Let $(H_1, H_2) \in \FF_p$ be a closest forbidden pair in the graph $G$ with $(J_1, J_2)$ being the vertex subsets corresponding to the pair. Let $P$ be a shortest path between $J_1$ and $J_2$. There is a family $\GG_1$ 
 such that 
\begin{itemize}
\item $\GG_1$ is a finite family of graphs of bounded-size (independent of the size of $G$), and
\item in the graph $G$ that is $\GG_1$-free, if a forbidden set $Q$ intersects the internal vertices of $P$, then $Q$ contains the right endpoint of $P$.
\end{itemize}

\end{enumerate}

\defparproblem{{\InfCMSOPiOnePiTwo}}{An undirected graph $G = (V, E)$, graph classes $\Pi_1, \Pi_2$ with associated forbidden families $\FF_1$ and $\FF_2$ such that Conditions \ref{2problem-condition-1} - \ref{2problem-condition-4} are satisfied and an integer $k$. }{$k$}{Is there a vertex set $S$ of size at most $k$ such that every connected component of $G - S$ is either in $\Pi_1$ or in $\Pi_2$?}


Note that the first two conditions for the problem are the same as those in {\PalphafreePiOnePiTwo}. Only the Condition \ref{2problem-condition-4} is changed which is tailored to prove the soundness of the branching rule we introduce. We also note that in the examples of {\InfCMSOPiOnePiTwo} we looked at (see Section \ref{sec:other-examples}), $\GG_1$ is the collection of all the graphs in $\ssp(\FF_1 , \FF_2)$ of some constant size, with the constant depending on the problem.

Towards an FPT algorithm for {\InfCMSOPiOnePiTwo},  We give the following branching rule whose soundness is easy to see.

\begin{branching rule}
\label{branch-rule:G1-free}
Let $(G, k)$ be the input instance and let $Q \subseteq V(G)$ such that $G[Q]$ is isomorphic to a graph in $\mathcal{G}_1$. Then, for each $v \in V(Q)$, delete $v$ from $G$ and decrease $k$ by 1. The resulting instance is $(G-v, k-1)$.
\end{branching rule}

From here on we assume that Branching Rule~\ref{branch-rule:G1-free} is not applicable for $G$ and so $G$ is  $\mathcal{G}_1$-free.
We now focus on connected components of $G$ which contain forbidden pairs. 
We have the following branching rule.

\begin{branching rule}
\label{branch-rule:H1-H2-pair-branching}
Let $(J^*, T^*)$ be the vertex subsets of a  closest forbidden pair $(H_1, H_2) \in \mathcal{F}_p$. Then for each $v \in J^* \cup T^*$, we delete $v$ and decrease $k$ by $1$, resulting in the instance $(G-v,k -1)$.
\end{branching rule}

We now prove the correctness of the above branching rule. 

\begin{lemma}
\label{lemma:nearest-H1-H2-pair-property}
Branching Rule \ref{branch-rule:H1-H2-pair-branching} is safe.
\end{lemma}

\begin{proof}
Let $P^*$ be a shortest path between $J^*$ and $T^*$ with endpoints $u \in J^*$ and $v \in T^*$. Since $G[J^*]$ and $G[T^*]$ cannot occur in the same connected component after deleting the solution, the solution must intersect $J^* \cup T^* \cup P^*$. 
We now prove that there exists an optimal solution of $(G, k)$ that does not intersect the internal vertices of $P^*$.

Let $X$ be an optimal solution such that $X \cap (J^* \cup T^*) = \emptyset$. Since $X \cap (J^* \cup T^* \cup P^*) \neq \emptyset$, $X$ must intersect the internal vertices of $P^*$. We claim that $X' = (X \setminus (P^* \setminus \{u\})) \cup \{v\}$ is also an optimal solution for $G$. Suppose not. Then there exists a forbidden set $Q$ such that $X' \cap Q = \emptyset$. Since $X \cap Q \neq \emptyset$, we know that $Q$ intersects the internal vertices of $P^*$. But then by Condition \ref{2problem-condition-4}, we know that $Q$ contains $v$ as well giving a contradiction.
\end{proof}

We now give the FPT algorithm {\InfCMSOPiOnePiTwo} which is the same algorithm in Theorem \ref{theorem:finite-path-FPT-algorithm}, but the Branching Rule \ref{branch-rule:closest-pair-branching-path} is replaced by Branching Rule~\ref{branch-rule:G1-free} and Branching Rule~\ref{branch-rule:H1-H2-pair-branching} in sequence. The correctness comes from Lemmas \ref{lemma:nearest-H1-H2-pair-property} and \ref{lemma:path-component-with-no-H1}.

 Again we define $f(k) = \max \{f_1(k), f_2(k)\}$ where $f_i(k)poly(n)$ is the running time for the algorithm $\mathcal{A}_i$. Also let $c$ be the maximum among the size of graphs in $\mathcal{G}_1$ and the integer $\max_{(H_1, H_2) \in \mathcal{F}_p}(|H_1|+|H_2|)$. Note that the branching factor of Branching Rule~\ref{branch-rule:H1-H2-pair-branching} is reduced to $\max_{(H_1, H_2) \in \mathcal{F}_p}(|H_1|+|H_2|)$ as we do not branch on the vertices of the path between the vertex sets of the closest forbidden pair.

\begin{theorem}
\label{theorem:H1-H2-FPT-algorithm}
{\InfCMSOPiOnePiTwo} can be solved in \\ ${\max\{f(k),c^k\}poly(n)}$-time.
\end{theorem}

We now give an approximation algorithm for {\InfCMSOPiOnePiTwo} when for $i \in \{1,2\}$, {\sc $\Pi_i$ Vertex Deletion} has  an approximation algorithm with approximation factor $c_i$. The algorithm is similar to that of Theorem \ref{theorem:finite-path-approximation-algorithm} with an additional primary step of greedily adding vertex subsets of induced graphs isomorphic to members in the family $\GG_1$ to the solution. Note that any optimal solution should contain at least one of the vertices of each such vertex subset. 


\begin{theorem}
\label{theorem:H1-H2-approximation-algorithm}
{\InfCMSOPiOnePiTwo} has a $d$-approximation algorithm where $d=\max \{c, c_1, c_2\}$.
\end{theorem}
\begin{proof}
Let $G$ be the input graph. Let $S_{OPT}$ be the  minimum sized set such that in the graph $G - S_{OPT}$, every connected component is either in $\Pi_1$ or $\Pi_2$. Let $|S_{OPT}| = OPT$.

Let $\mathcal{S}_1$ denote the maximal family of graphs which are in $\mathcal{G}_1$ such that any two members of $\mathcal{S}_1$ are pairwise disjoint. We can greedily construct such a family $\mathcal{S}_1$ in polynomial time. Let $S_1$ be the set of vertices that is present in any graphs in $\mathcal{S}_1$. From Lemma \ref{lemma:forbidden-pair-characterization}, we can conclude that any feasible solution of $G$ must contains a vertex from each member of the family $\mathcal{S}_1$. Since the members of $\mathcal{S}_1$ are pairwise disjoint, we have that $|S_{OPT} \cap S_1| \geq |\mathcal{S}_1|$.

Let $G' = G - S_1$. We now construct a family $\mathcal{S}_2$ as follows. Initially $\mathcal{S}_2 = \emptyset$. We find the closest forbidden pair $(J^*,T^*)$ in $G'$, add it to $\mathcal{S}_2$ and delete $J^* \cup T^*$ from $G'$. We repeat this step until it is no longer applicable. Let $S_2$ be the set of vertices that is present in any pair of graphs in $\mathcal{S}_2$. From Lemma \ref{lemma:nearest-H1-H2-pair-property}, we can conclude that any feasible solution of $G$ must contains a vertex from each pair of the family $\mathcal{S}_2$. Since the members of $\mathcal{S}_2$ are pairwise disjoint and $S_{OPT} - S_1$ is also an optimum solution for $G'$, we have that $|(S_{OPT}-S_1) \cap S_2| \geq |\mathcal{S}_2|$.

Let $G'' = G' - S_2$. We now construct a set $S_3$ as follows. Let $C_1, \dotsc , C_q$ be the connected components of $G''$. If a connected component $C_i$ has no graphs in $\mathcal{F}_p^1$ as induced subgraph, we apply the $c_1$-approximation algorithm for $\Pi_1$ vertex deletion on $G''[C_i]$ to obtain a solution $S_i$. If a connected component $C_i$ no has no graphs in $\mathcal{F}_p^2$ as induced subgraph, we apply the $c_2$-approximation algorithm for $\Pi_2$ vertex deletion on $G''[C_i]$ to obtain a solution $Z_i$. The correctness comes from Lemma \ref{lemma:path-component-with-no-H1}. We have $S_3 = \bigcup_{i \in [q]} Z_i$. Since $((S_{OPT} - S_1) -S_2) \cap C_i$ is a feasible solution for the connected component $C_i$ of $G''$, we have that $|Z_i| \leq (\max \{c_1,c_2\})|((S_{OPT} - S_1) -S_2) \cap C_i|$ for all $i \in [q]$.

We set $S = S_1 \cup S_2 \cup S_3$. We have 
\begin{eqnarray*}
|S|  &=&  |S_1| + |S_2| + |S_3| \\
 & \leq & (\max_{H \in \mathcal{G}_1} \{|H|\})|\mathcal{S}_1| + (\max_{(H_1, H_2) \in \mathcal{F}_p}(|H_1|+|H_2|))|\mathcal{S}_2| + \\
 & & \sum\limits_{i=1}^{q} (\max \{c_1,c_2\})|((S_{OPT} - S_1) -S_2) \cap C_i|\\
& \leq & c|S_{OPT} \cap S_1| + c|(S_{OPT}-S_1) \cap S_2| + \\ & &\sum\limits_{i=1}^{q} (\max \{c_1,c_2\})|((S_{OPT} - S_1) -S_2) \cap C_i|\\
& \leq & (\max \{c, c_1, c_2\})|S_{OPT}|
\end{eqnarray*}

Thus we have a $d$-approximation algorithm for {\InfCMSOPiOnePiTwo}.
\end{proof}

\section{Examples of {\InfCMSOPiOnePiTwo}}
\label{sec:other-examples}
Verifying whether Conditions \ref{2problem-condition-2} and \ref{2problem-condition-4} are satisfied for a general $\Pi_1$ and $\Pi_2$ is non-trivial. Hence we look at specific pairs of graph classes $\Pi_1$ and $\Pi_2$ and prove that they are examples of {\InfCMSOPiOnePiTwo}.

We start by showing that the problems {\ClawfreeorTriangleFree} 
is indeed an example of {\InfCMSOPiOnePiTwo}.

\begin{lemma} {\ClawfreeorTriangleFree} is an example of {\InfCMSOPiOnePiTwo}.
\end{lemma}
\begin{proof}
We show that Conditions \ref{2problem-condition-1} - \ref{2problem-condition-4}  are satisfied by {\ClawfreeorTriangleFree}. The problems {\sc Claw-Tree Vertex Deletion} and {\sc Triangle Vertex Deletion} has simple $4^k poly(n)$ and $3^k poly(n)$ time FPT algorithms parameterized by the solution size $k$ via a simple branching on claws and triangles respectively. Hence Condition \ref{2problem-condition-1} is satisfied by  {\ClawfreeorTriangleFree}. The forbidden pair family for {\ClawfreeorTriangleFree} is of size one which is $\{K_{1,3}, C_3\}$. Hence Condition \ref{2problem-condition-2} is satisfied. Finally, from Lemma \ref{lemma:branching-rule-claw-triangle-safeness}, we can conclude that Condition \ref{2problem-condition-4} is satisfied as well with $\GG_1 = \emptyset$. 

Hence {\ClawfreeorTriangleFree} is an example of {\InfCMSOPiOnePiTwo}. We have $f(k) = \max \{4^k, 3^k\} = 4^k$ and $c = \max_{(H_1, H_2) \in \mathcal{F}_p}(|H_1|+|H_2|) =7$. Hence from Theorem \ref{theorem:H1-H2-FPT-algorithm}, we have a  $7^k poly(n)$ time algorithm for {\ClawfreeorTriangleFree}. This is the same running time obtained independently in Theorem \ref{theorem:claw-free-triangle-free-FPT-algorithm}.

We also have $d=\max \{c, c_1, c_2\} =7$ giving a $7$-approximation for {\ClawfreeorTriangleFree} from Theorem \ref{theorem:H1-H2-approximation-algorithm}, which is the same approximation factor obtained independently in Theorem \ref{theorem:claw-free-triangle-free-approx-algorithm}.
\end{proof}

We now give examples of other pairs of graph classes $\Pi_1$ and $\Pi_2$ whose scattered deletion problem is an example of {\InfCMSOPiOnePiTwo}. The core part in each of the cases below is establishing that Condition \ref{2problem-condition-4} is satisfied. We do so by establishing structural properties for the shortest path corresponding to the closest forbidden pair. Such properties vary for each case.

\subsection{Interval or Trees}

We define the problem. 

\defparproblem{{\IVDTrees}}{An undirected graph $G = (V, E)$ and an integer $k$.}{$k$}{Is there $S \subseteq V(G)$ of size at most $k$ such that every connected component of $G - S$ is either an interval graph, or a tree?}

We have the following forbidden subgraph characterization of interval graphs.

\begin{lemma}(\cite{lekkeikerker1962representation})
\label{lemma:interval-main-property}
A graph is an interval graph if and only if it does not contain net, sun, hole, whipping top, long-claw, $\dagger$-AW, or $\ddagger$-AW as its induced subgraphs. 
\end{lemma}
See Figure~\ref{fig:interval-obstruction} for an illustration of the graphs mentioned as forbidden subgraphs for interval graphs. Note that $\dagger$-AW and $\ddagger$-AW are a collection of graphs (like holes) and its size need not be constant.

\begin{figure}[t]
\centering
	\includegraphics[scale=0.25]{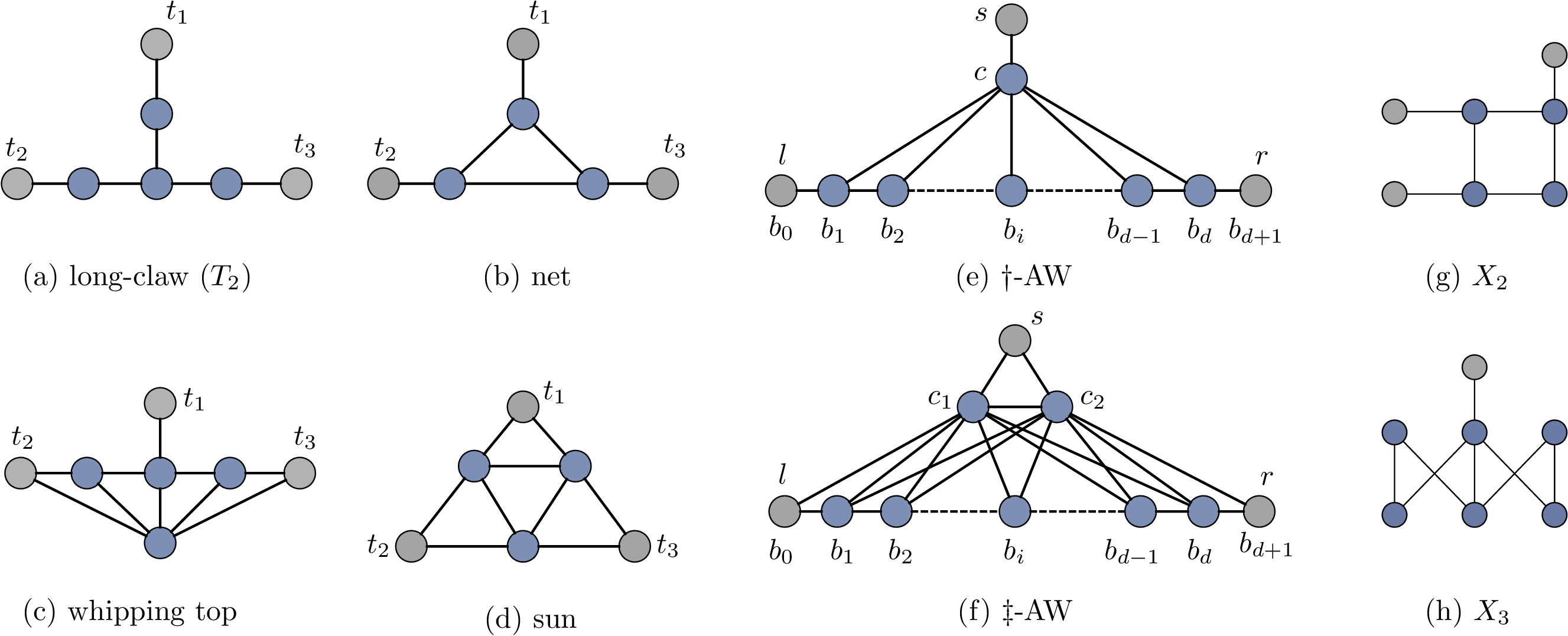}
	\caption{Obstructions for Graph Classes}
\label{fig:interval-obstruction}	
\end{figure}

We now give a characterization for graphs whose every connected component is an interval graph or a tree.
Recall from Section \ref{sec:forbidden-characterization} that the forbidden pair family $\FF_p$ of this pair of graph classes is $($long-claw, triangle$)$ and $\ssp(\FF_1 , \FF_2)$ is  $\{$net, sun, hole, whipping top, $\dagger$-AW, $\ddagger$-AW$\}$. 
The following is a corollary from Lemma \ref{lemma:forbidden-pair-characterization}.
\begin{corollary}
\label{lemma:interval-or-tree-some-characteristics}
The following statements are equivalent.
\begin{enumerate}
	\item\label{char-item-1} Let $G$ be a graph such that every connected component is either an interval graph or a tree.
	\item\label{char-item-2} $G$ does not have any net, sun, hole, whipping top, $\dagger$-AW, $\ddagger$-AW as its induced subgraphs. Moreover, $G$ cannot have long-claw and triangle as induced subgraphs in the same connected component. 
\end{enumerate}
\end{corollary}

We show that Conditions \ref{2problem-condition-1} - \ref{2problem-condition-4}  are satisfied by {\IVDTrees}. The problems {\sc Interval Vertex Deletion} and {\sc Feedback Vertex Set} has  $10^kpoly(n)$ \cite{cao2015interval} and $3.619^k poly(n)$ \cite{KP14} time FPT algorithms parameterized by the solution size $k$ respectively. Hence Condition \ref{2problem-condition-1} is satisfied by  {\IVDTrees}. The forbidden pair family for {\IVDTrees} is of size one which is the pair $($long-claw, triangle$)$. Hence Condition \ref{2problem-condition-2} is satisfied.

It remains to show that Condition \ref{2problem-condition-4} is satisfied for {\IVDTrees}. We define $\mathcal{G}_1$ be the family of graphs in $\ssp(\FF_1 , \FF_2)$ of size at most 10.

Let $(J^*, T^*)$ be the vertex subsets of a closest long-claw, triangle pair in a connected component of $G$, where $J^*$ is a long-claw and $T^*$ is a triangle. Let $P^* := x_0,x_1,\ldots,x_{d-1},x_d$ be a shortest path between $J^*$ and $T^*$ of length $d_G(J^*, T^*) = d$ with $x_0 = u \in J^*$ and $x_d = v \in T^*$.

A {\em caterpillar} graph is a tree in which all the vertices are within distance $1$ of a central path. In the graph $G$, let $C$ be the connected component of $G - (J^* \cup T^*)$ containing the internal vertices of $P^*$.  We have the following lemma that helps us to prove Condition \ref{2problem-condition-4}.

\begin{lemma}
\label{lemma:nearest-long-claw-triangle-pair-path-neighbors}
The graph $C$ is a caterpillar with the central path being $P^*$. Furthermore, the only vertices of $C$ adjacent to $J^* \cup T^*$ are $x_1$ and $x_{d-1}$ which are only adjacent to $x_0$ and $x_d$ respectively.
\end{lemma}

\begin{proof}
\begin{figure}[t]
\centering
	\includegraphics[scale=0.28]{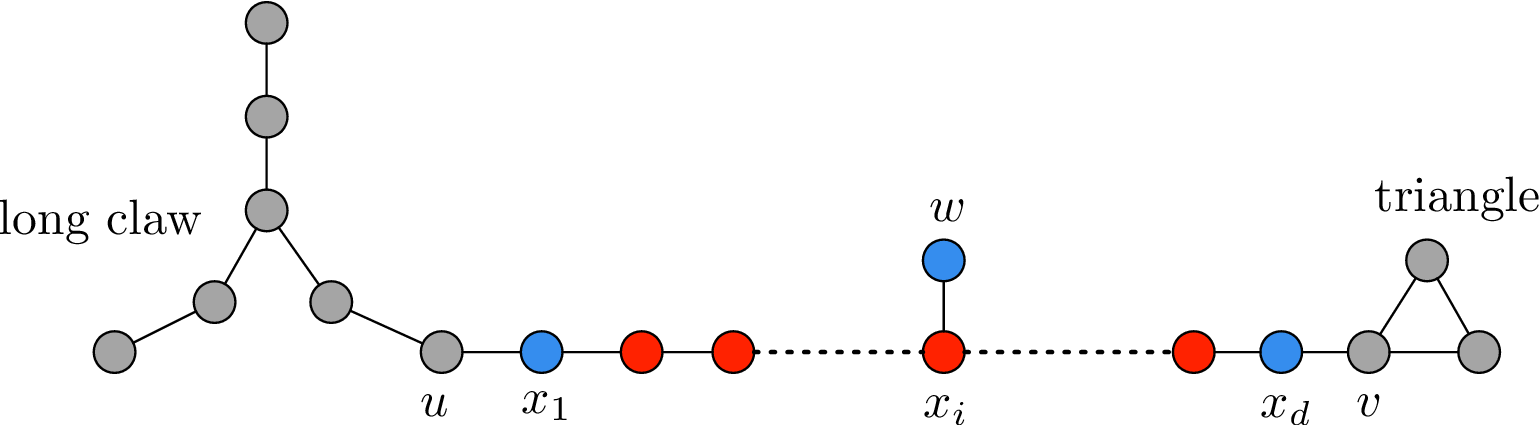}
	\caption{An illustration of a shortest path between a closest (long-claw, triangle) pair.}
\label{fig:long-claw-triangle-pair}
\end{figure}

We first look at the neighborhood vertices of the path $P^*$ in the connected component $C$. Let $w$ be such a vertex 
which is adjacent to a vertex $x_i$ with $ i \in \{1,2, \dotsc, d-1\}$. We prove that  $w$ is not adjacent to any other vertex in $G$. Thus there are no cycles in $C$ and all the vertices in $C$ are at distance at most $1$ to $P^*$ proving that $C$ is a caterpillar.

We go over the possibilities of edges from $w$ to other vertices.

\begin{description}
	\item[Case 1:] Suppose $w$ is adjacent to some vertex $x_j$ with $j \in \{0, \dotsc , d\}$. If $j=i+1$ or $j=i-1$, then $wx_ix_j$ forms a triangle $T'$. Since the path $P' = u,x_1,\ldots,x_{i}$ has length smaller than $P^*$ and connects $J^*$ and $T'$, we have that $d_G(J^*, T') < d_G(J^*, T^*)$. This contradicts the fact that $(J^*, T^*)$ is a pair of long-claw and triangle that is the closest.

Hence $w$ is not adjacent to $x_{i-1}$ and $x_{i+1}$. Suppose $j=i+2$ or $j=i-2$. Then the graph induced by the set of vertices $\{w, x_i, x_{i+1}, x_{i+2}\}$ or $\{w, x_i, x_{i-1}, x_{i-2}\}$ is $C_4$ contradicting that the graph is $\GG_1$-free.
Hence $w$ is not adjacent to $x_{i-2}$ and $x_{i+2}$.

Suppose $w$ is adjacent to $x_j$ with  $1 \leq j<i-2$ or $i+2 < j \leq d$. Then note that we have a path from $x_j$ to $x_i$ of length 2 via $w$ which is shorter than the path $x_i \dotsc x_j$ along $P^*$. Hence we have a path $P' = u x_1 \dotsc x_j w x_i \dotsc x_dv$ or $P' = u x_1 \dotsc x_i w x_j \dotsc x_dv$ from $J^*$ to $T^*$ of length smaller than $P^*$ contradicting that $P^*$ is the shortest such path. 

Hence $w$ is not adjacent to any of the vertices in $P^*$.

	\item[Case 2:] Suppose $w$ is adjacent to a vertex $u' \in J^*$. We assume without loss of generality that among all neighbors of $w$ in $J^*$, $u'$ is the vertex that is closest to $u$ in $J^*$, i.e, $d_{J^*}(u,u')$ is minimum. If $3 \leq i \leq d-1$, we have a path from $u'$ to $x_i$ of length 2 via $w$ which is shorter than the path from $u$ to $x_i$ via $P^*$. This contradicts that $P^*$ is the shortest path from $J^*$ to $T^*$. 
	
Suppose $i=1$ or $i=2$. Let $P'$ denote a shortest path between $u$ and $u'$ in $J^*$. Since $J^*$ is a long-claw, the length of $P'$ is at most $4$. Let us concatenate $P'$ with the prefix of the path $P^*$ from $u$ to $x_i$ which is of length at most $2$. Then we get a shortest path from $u'$ to $x_i$ which is of length at least $2$ and at most $6$. The vertex $w$ is adjacent to only $u'$ and $x_i$ in this path. Hence the graph induced by the set of vertices in this path plus $w$ is cycle $C_j$ with $4 \leq j \leq 8$ contradicting that the graph is $\GG_1$-free.

Hence $w$ is not adjacent to any of the vertices in $J^*$.

	\item[Case 3:] Suppose $w$ adjacent to a vertex $v' \in T^*$. We have $d_{T^*}(v,v') = 1$ as $T^*$ is a triangle. If $1 \leq i \leq d-3$, we have a path from $v'$ to $x_i$ of length 2 via $w$ which is shorter than the path from $v$ to $x_i$ via $P^*$. This contradicts that $P^*$ is the shortest path from $J^*$ to $T^*$. 
	
	Hence $d-2 \leq i \leq d-1$. Let $P'$ be the suffix of the path $P^*$ from $x_i$ to $v$. Then we get a shortest path from $x_i$ to $v'$ as $x_i P' vv'$ which is of length at least $2$ and at most $3$. The vertex $w$ is adjacent with only vertices $x_i$ and $v'$ in this path. Hence the graph induced by the set of vertices of this path plus $w$ forms a cycle $C_j$ with $4 \leq j \leq 5$ contradicting that the graph is $\GG_1$-free.

Hence from the above three cases, $w$ is adjacent to none of the other vertices of $J^* \cup T^* \cup P^*$.

	\item[Case 4:] We now prove that $w$ is not adjacent to any other vertex $w' \in V(C) \setminus P^*$. Suppose that there exists such a vertex $w'$. We now look at various cases of the adjacency of $w'$ with other vertices.
	\begin{description} 
		\item[- Case 4.1:] We first look at adjacencies of $w'$ with vertices in $P^*$.

		Suppose $w'$ is adjacent to vertex $x_i$. Then the graph induced by the set of vertices $T' = \{w', w, x_i\}$ forms a triangle.  Since the path $P' := u, x_1,\ldots, x_{i}$ has length smaller than $P^*$ and connects $J^*$ and $T'$, we have that $d_G(J^*, T') < d_G(J^*, T^*)$. This contradicts the fact that $(J^*, T^*)$ is a pair of long-claw and triangle that is the closest.

		Suppose $w'$ is adjacent to $x_{i+1}$ or $x_{i-1}$. Then the graph induced by the set of vertices $\{w',w, x_i, x_{i+1}\}$ or $\{w', w, x_i, x_{i-1}\}$ forms a $C_4$ contradicting that the graph is $\GG_1$-free.
		

		Hence this is not the case. Now suppose $w'$ is adjacent to $x_{i+2}$ or $x_{i-2}$. Then the graph induced by the set of vertices $\{w', w, x_i, x_{i+1}, x_{i+2}\}$ or $\{w',w, x_i, x_{i-1}, x_{i-2},\}$ forms a $C_5$ contradicting that the graph is $\GG_1$-free.

		Hence this is also not the case. Suppose $w'$ is adjacent to $x_{i+3}$ or $x_{i-3}$. Then the graph induced by the set of vertices $\{w', w, x_i, x_{i+1}, x_{i+2}, x_{i+3}\}$ or $\{w', w, x_i, x_{i-1}, x_{i-2}, x_{i-3}\}$ forms a $C_6$ contradicting that the graph is $\GG_1$-free.

		 In the remaining case, $w'$ is adjacent to $x_j$ with $1 \leq j<i+3$ or $i+3 < j \leq d$. Hence $|j-i| >3$. Then note that the path $x_j,w',w, x_i$ is of length three which is shorter than the path between $x_i$ and $x_j$ in $P^*$. Hence we get a path $P'$ from $u$ to $v$ of length smaller than $P^*$ contradicting that $P^*$ is the smallest such path. Hence $w'$ is not adjacent to any of the vertices in $P^*$. 
		 

		\item[- Case 4.2:] Suppose $w'$ adjacent to a vertex $u' \in J^*$. We assume without loss of generality that among all neighbors of $w$ in $J^*$, $u'$ is the vertex that is closest to $u$ in $J^*$, i.e, $d_{J^*}(u,u')$ is minimum. If $3 \leq i \leq d-1$, observe that the path $u'w'wx_i$ from $u'$ to $x_i$ is of length 3 which is shorter than the path from $u$ to $x_i$ via $P^*$. This contradicts that $P^*$ is the shortest path from $J^*$ to $T^*$. Hence $1 \leq i \leq 3$. Let $P'$ denote the path between $u$ and $u'$ in $J^*$ and $P''$ be the prefix of the path $P^*$ from $u$ to $x_i$. Then $u P'' x_i ww' u' P'u$ forms a cycle $C_j$ with $4 \leq j \leq 10$ with no chords contradicting that the graph is $\GG_1$-free.

		Hence $w$ is not adjacent to any of the vertices in $J^*$. 

		\item[- Case 4.3:] Suppose $w'$ adjacent to a vertex $v' \in T^*$. We have $d_{T^*}(v,v') = 1$ as $T^*$ is a triangle. If $1 \leq i \leq d-4$, we have the path $v',w',w,x_i$ from $v'$ to $x_i$ of length 3 which is shorter than the path from $v$ to $x_i$ via $P^*$. This contradicts that $P^*$ is the shortest path from $J^*$ to $T^*$. Hence $d-3 \leq i \leq d-1$. Let $P'$ be the suffix of the path $P^*$ from $x_i$ to $v$. Then $v'w'wx_i P' vv'$ forms a cycle $C_j$ with $4 \leq j \leq 6$ without any chords contradicting that the graph is $\GG_1$-free.
	\end{description}
\end{description}

Hence we conclude that $w'$ is not adjacent to any of the vertices in $J^* \cup T^* \cup P^*$. Now look the graph induced by the set of vertices $J'$ which is 
\begin{itemize}
\item $\{w', w, x_i, x_{i-1}, x_{i-2}, x_{i+1}, x_{i+2} \}$ for $3 \leq i \leq x_{d-2}$ or 
\item $\{w', w, x_i, x_{i-1}, u, x_{i+1}, x_{i+2} \}$ when $i=2$ or 
\item $\{w', w, x_i, u, u', x_{i+1}, x_{i+2}\}$ when $i=1$ for $u' \in J^* \cap N(u)$ or 
\item $\{w', w, x_i, x_{i-1}, x_{i-2}, v, v' \}$ when $i=d-1$ for $v' \in J^* \cap N(v)$.
\end{itemize}
 In all cases, the graph induced by $J'$ forms a long-claw. Since the path $P'$ from $J'$ to $v \in T^*$ has length smaller than $P^*$, we have that $d_G(J', T^*) < d_G(J^*, T^*)$. This contradicts the fact that $(J^*, T^*)$ is a pair of long-claw and triangle that is closest. 

Hence no such vertex $w'$ exists and therefore $w$ has no other neighbors in $G$.

Hence, the graph $C$ is a caterpillar with the central path being $P^*$. Furthermore, no vertices other than $x_1$ and $x_{d-1}$ is adjacent to $J^* \cup T^*$.
\end{proof}

We now use Lemma \ref{lemma:nearest-long-claw-triangle-pair-path-neighbors} to prove that Condition \ref{2problem-condition-4} is satisfied for {\IVDTrees}.

\begin{lemma}
\label{lemma:branching-rule-2-safeness}
 Condition \ref{2problem-condition-4} is satisfied for {\IVDTrees}.
\end{lemma}

\begin{proof}
Condition \ref{2problem-condition-4} is not satisfied in the following case. There exist a pair  $(J^*, T^*)$ which is the vertex subsets of a closest long-claw, triangle pair in a connected component of $G$, where $J^*$ is a long-claw and $T^*$ is a triangle. Also there is a shortest path $P^* := x_0,x_1,\ldots,x_{d-1},x_d$ between $J^*$ and $T^*$ of length $d_G(J^*, T^*) = d$ with $x_0 = u \in J^*$ and $x_d = v \in T^*$. A forbidden set $Q$ of the graph $G$ is such that $Q$ contains some internal vertex $x_i$ of the path $P$ but it does not contain $v$.

Since the graph $G$ is $\GG_1$-free, $G[Q]$ can be one of hole, $\dagger$-AW or a $\ddagger$-AW or contain a forbidden pair for $($long-claw, triangle$)$.
Note that all of these possibilities contain cycles. But  from Lemma \ref{lemma:nearest-long-claw-triangle-pair-path-neighbors}, the component $C$ of $G \setminus (J^* \cup  T^*)$ that contains the internal vertices of $P^*$ is a caterpillar which does not contain any cycles. Hence $Q$ is not fully contained in $C$.

From Lemma \ref{lemma:nearest-long-claw-triangle-pair-path-neighbors}, the only neighbors of $C$ are $u$ and $v$ via $x_1$ and $x_{d-1}$ respectively.
Since $G[Q]$ is connected and intersects $x_i$, $Q \cap \{u, v\} \neq \emptyset$. Since we assumed that $Q$ does not contain the vertex $v$, we have $u \in Q$. In particular, $Q$ contains the entire subpath of $P^*$ from $u$ to $x_i$.

Also note that $u$ cannot be part of a subset of three vertices $T'$ which is a triangle as otherwise, we get a pair $(J^*, T')$ with distance zero contradicting that $(J^*, T^*)$ was the closest pair $P^*$ has internal vertices. Hence the forbidden set $Q$ cannot be $\dagger$-AW or a $\ddagger$-AW whose structure forces $u$ to be part of a triangle if it contains $x_i$ (which also happens only in the case when $i=1$). Since $x_i$ does not have any paths to the vertex $u$ other than the subpath in $P^*$, the forbidden set $Q$ cannot be a hole as well. 

Hence $Q$ can only correspond to a $($long-claw, triangle$)$ forbidden pair. In this case, we claim that the set after removing the vertices $x_1, \dotsc, x_i$ from $Q$ is also a forbidden set. This contradicts that $Q$ is a forbidden set as by definition they are required to be minimal.

Since $Q$ is a forbidden set, it contains vertex subsets that are isomorphic to a long-claw and a triangle. From Lemma \ref{lemma:nearest-long-claw-triangle-pair-path-neighbors}, we can conclude that none of the vertices $x_1, \dotsc , x_i$ can be part of any triangle in $G$. None of these vertices can be part of a long-claw in $G$ as well since it contradicts that $(J^*, T^*)$ is the closest forbidden pair. Since $v$ disconnects the path $P^*$, the subpath $x_1, \dotsc , x_i$ is not part of a path connecting a long-claw and a triangle either as if so $Q$ must contain the entire path $P^*$ including $v$. Hence the set after removing the vertices $x_1, \dotsc , x_i$ from $Q$ is still a forbidden set contradicting that $Q$ is minimal.

These cases of $Q$ are mutually exhaustive completing the proof of the Lemma.
\end{proof}

Hence, we have established that {\IVDTrees} is indeed an example of {\InfCMSOPiOnePiTwo}. We have the following theorem.

\begin{theorem} {\IVDTrees} has an FPT algorithm with running time $10^k poly(n)$ and a $10$-approximation algorithm.
\end{theorem}

\begin{proof}
We have $f(k) = \max \{10^k, 3.619^k\} = 10^k$ and $c = \max_{(H_1, H_2) \in \mathcal{F}_p}(|H_1|+|H_2|) =10$. Hence from Theorem \ref{theorem:H1-H2-FPT-algorithm}, we have a  $10^k poly(n)$ time algorithm for {\IVDTrees}. 

We know that {\sc Interval Vertex Deletion} has an $8$-approximation algorithm \cite{cao2016linear} and {\sc Feedback Vertex Set} has a 2-approximation algorithm \cite{bafna19992}. Hence $d=\max \{c, c_1, c_2\} =10$ giving a $10$-approximation for {\IVDTrees} from Theorem \ref{theorem:H1-H2-approximation-algorithm}.
\end{proof}

\subsection{Proper Interval or Trees}
We define the problem.

\defparproblem{{\PIVDTrees}}{An undirected graph $G = (V, E)$ and an integer $k$}{$k$}{Is there $S \subseteq V(G)$ of size at most $k$ such that every connected component of $G - S$ is a proper interval graph or a tree?}

We have the following forbidden subgraph characterization of proper interval graphs.

\begin{lemma}\cite{brandstadt1999graph}
\label{lemma:proper-interval-main-property}
A graph is said to be a proper interval graph if and only if it does not contain claw, net, sun or hole as its induced subgraphs. 
\end{lemma}

We now give a characterization for graphs whose every connected component is a proper interval graph or a tree.

\begin{lemma}
\label{lemma:proper-interval-or-tree-some-characteristics}
The following statements are equivalent.
\begin{enumerate}
	\item\label{pint:char-item-1} A graph $G$ is such that every connected component of $G$ is a proper interval graph or a tree.
	\item\label{pint:char-item-2} A graph $G$ does not have any net, sun or hole as its induced subgraphs. Moreover, no connected component of $G$ have a claw and a triangle as induced graphs.
\end{enumerate}
\end{lemma}
\begin{proof}
We prove that the forbidden pair family is $($claw, triangle$)$. The forbidden family $\FF_1$ for proper interval graphs are claw, net, sun or holes. The forbidden family of trees $\FF_2$ is cycles. Since cycles of length at least 4 are common in $\FF_1$ and $\FF_2$, they are in $\ssp(\FF_1 , \FF_2)$. Since the graphs net and sun have triangle as induced subgraph, they are in $\ssp(\FF_1 , \FF_2)$ as well. The only remaining pair in $\FF_1 \times \FF_2$ is $($claw, triangle$)$. The proof now follows from the forbidden characterization in Lemma \ref{lemma:forbidden-pair-characterization}.
 \end{proof}

Hence Condition \ref{2problem-condition-2} is satisfied by {\PIVDTrees} as the forbidden pair is of size one which is $($claw, triangle$)$. 
Since {\sc Proper Interval Vertex Deletion} is FPT with a $O^*(6^k)$ running time algorithm from \cite{van2013proper} and  {\sc Feedback Vertex Set } is FPT with a $O^*(3.619^k)$ running time algorithm from \cite{KP14}, Condition \ref{2problem-condition-1} is satisfied as well.

We now prove that Condition \ref{2problem-condition-4} is satisfied by {\PIVDTrees}. Recall Lemma \ref{lemma:nearest-claw-triangle-pair-path} where we established that the internal vertices of any shortest path between the vertex sets of a closest claw, triangle pair in a graph do not contain any neighbors other than the endpoints of the path.

\begin{lemma}Condition \ref{2problem-condition-4} is satisfied by {\PIVDTrees}.
\end{lemma}
\begin{proof}

Condition \ref{2problem-condition-4} is not satisfied in the following case. There exist a pair  $(J^*, T^*)$ which is the vertex subsets of a closest claw, triangle pair in a connected component of $G$, where $J^*$ is a claw and $T^*$ is a triangle. Also there is a shortest path $P^* := x_0,x_1,\ldots,x_{d-1},x_d$ between $J^*$ and $T^*$ of length $d_G(J^*, T^*) = d$ with $x_0 = u \in J^*$ and $x_d = v \in T^*$. A forbidden set $Q$ of the graph $G$ is such that $Q$ contains some internal vertex $x_i$ of the path $P$ but it does not contain $v$.

Let $\GG_1 = \emptyset$. The graph $G[Q]$ can be one of net, sun, hole or contain a forbidden pair for $($claw, triangle$)$. Note that all of these possibilities contain cycles. But from Lemma \ref{lemma:nearest-claw-triangle-pair-path}, the component $C$ of $G \setminus (J^* \cup  T^*)$ that contains the internal vertices of $P^*$ is the path $P^* \setminus \{u,v\}$  which does not contain any cycles. Hence $Q$ is not fully contained in $C$.

From Lemma \ref{lemma:nearest-claw-triangle-pair-path}, the only neighbors of $C$ is $u$ and $v$ via $x_1$ and $x_{d-1}$ respectively.  Since $G[Q]$ is connected and intersects $x_i$, $Q \cap \{u, v\} \neq \emptyset$. Since we assumed that $Q$ does not contain the vertex $v$, we have $u \in Q$. In particular, $Q$ contains the entire subpath of $P^*$ from $u$ to $x_i$.

Also note that $u$ cannot be part of a subset of three vertices $T'$ which is a triangle as otherwise, we get a pair $(J^*, T')$ with distance zero contradicting that $(J^*, T^*)$ was the closest pair $P^*$ has internal vertices. Hence the forbidden set $Q$ cannot be a net or a sun whose structure forces $u$ to be part of a triangle if it contains $x_i$. Since $x_i$ does not have any paths to the vertex $u$ other than the subpath in $P^*$, the forbidden set $Q$ cannot be a hole as well. 

Hence $Q$ can only correspond to a $($claw, triangle$)$ forbidden pair. In this case, we claim that the set after removing the vertices $x_1, \dotsc, x_i$ from $Q$ is also a forbidden set. This contradicts that $Q$ is a forbidden set as by definition they are required to be minimal.

Since $Q$ is a forbidden set, it contains vertex subsets that are isomorphic to a claw and a triangle. From Lemma \ref{lemma:nearest-claw-triangle-pair-path}, we can conclude that none of the vertices $x_1, \dotsc , x_i$ can be part of any triangle in $G$. None of these vertices can be part of a claw in $G$ as well since it contradicts that $(J^*, T^*)$ is the closest forbidden pair. Since $v$ disconnects the path $P^*$, the subpath $x_1, \dotsc , x_i$ is not part of a path connecting a claw and a triangle either as if so $Q$ must contain the entire path $P^*$ including $v$. Hence the set after removing the vertices $x_1, \dotsc , x_i$ from $Q$ is still a forbidden set contradicting that $Q$ is minimal.

These cases of $Q$ are mutually exhaustive completing the proof of the Lemma.
\end{proof}

We have $f(k) = \max \{6^k, 3.619^k\} = 6^k$ and $c = 7$. 
Also we know that {\sc Proper Interval Vertex Deletion} has an 6-approximation algorithm \cite{van2013proper} and {\sc Feedback Vertex Set} has a 2-approximation algorithm \cite{bafna19992}. Hence we have $d=7$ as well. We have the following theorem.

\begin{theorem}
\label{theorem:proper-interval-trees-algorithms}
{\PIVDTrees} can be solved in $7^k poly(n)$-time and has a $7$-approximation algorithm.
\end{theorem}

\subsection{Chordal or Bipartite Permutation}

We define the problem as follows.

\defparproblem{{\ChordalBipPer}}{An undirected graph $G = (V, E)$ and an integer $k$}{$k$}{Is there $S \subseteq V(G)$ of size at most $k$ such that every connected component of $G - S$ is either a chordal graph, or a bipartite permutation graph?}

The forbidden set for chordal graphs $\FF_1$ is the set of cycle graphs with a length of at least $4$. We have the following characterization for bipartite permutation graphs which defines $\FF_2$.

\begin{lemma}\cite{bozyk2020vertex}
\label{lemma:bipartite-permutation-main-property}
A graph is said to be a bipartite permutation graph if and only if it does not contain long-claw,$ X_2, X_3, C_3$ or cycle graphs of length at least $5$ as its induced subgraphs. 
See Figure~\ref{fig:interval-obstruction} for an illustration of the graphs $X_2$ and $X_3$.
\end{lemma}

We now give a characterization for graphs whose each connected component is either a chordal graph or a bipartite permutation graph.

\begin{lemma}
\label{lemma:chordal-or-bip-per-some-characteristics}
The following statements are equivalent.
\begin{enumerate}
	\item\label{chord-bip-char-item-1} Let $G$ be a graph such that every connected component is either chordal or a bipartite permutation graph.
	\item\label{chord-bip-char-item-2} $G$ does not have any $X_2,X_3$ or induced cycle of length at least $5$ as induced subgraphs. Moreover, $G$ cannot have long-claw and $C_4$ in the same connected component or have $C_4$ and triangle in the same connected component. 
\end{enumerate}
\end{lemma}

\begin{proof}
We prove that the forbidden pair family is $(C_4$, long-claw$)$ and  $(C_4, C_3 )$. The forbidden family $\FF_1$ for chordal graphs are holes. The forbidden family of bipartite permutation graphs $\FF_2$ is long-claw,$ X_2, X_3, C_3$  plus cycles of length at least $5$. Since cycles of length at least $5$ are common in $\FF_1$ and $\FF_2$, they are in $\ssp(\FF_1 , \FF_2)$. Since the graphs $X_2, X_3$ have $C_4$ as induced subgraph, they are in $\ssp(\FF_1 , \FF_2)$ as well. The only remaining pairs in $\FF_1 \times \FF_2$ are  $(C_4$, long-claw$)$ and  $(C_4$, triangle$)$. The proof now follows from the forbidden characterization in Lemma \ref{lemma:forbidden-pair-characterization}.
\end{proof}

Hence Condition \ref{2problem-condition-2} is satisfied by {\ChordalBipPer} as the forbidden pair is of size two which are $($long-claw, $C_4)$ and $($triangle$,C_4)$. 
Since {\sc Chordal Vertex Deletion} is FPT with a $k^{O(k)}poly(n)$ running time algorithm from \cite{cao2016chordal} and  {\sc Bipartite Permutation Vertex Deletion Set } is FPT with a $9^k poly(n)$ running time algorithm from \cite{bozyk2020vertex}, Condition \ref{2problem-condition-1} is satisfied as well.

It remains to show that Condition \ref{2problem-condition-4} is satisfied by {\ChordalBipPer}. We define $\mathcal{G}_1$ as all the forbidden graphs in $\ssp(\FF_1 , \FF_2)$ of size at most $10$. 

Let $(J^*, T^*)$ be the vertex subsets of a closest forbidden pair in a connected component of a $\GG_1$-free $G$, where $J^*$ is one of long-claw or triangle and $T^*$ is a $C_4$. Let $P^* := x_0,x_1,\ldots,x_{d-1},x_d$ be a shortest path between $J^*$ and $T^*$ of length $d_G(J^*, T^*) = d$ with $x_0 = u \in J^*$ and $x_d = v \in T^*$. 

Let $C$ be the connected component of $G - (J^* \cup T^*)$ containing the internal vertices of $P^*$.  We have the following lemma similar to Lemma \ref{lemma:nearest-long-claw-triangle-pair-path-neighbors} in {\IVDTrees}.

\begin{lemma}
\label{lemma:nearest-claw-triangle-C4-pair-path-neighbors}
The graph $C$ is a caterpillar with the central path being $P^* \setminus \{u,v\}$. Furthermore, the only vertices of $C$ adjacent to $J^* \cup T^*$ are $x_1$ and $x_{d-1}$ which are only adjacent to $x_0$ and $x_d$ respectively.
\end{lemma}
\begin{proof}

Let $w$ be a vertex of $C$ other than $P^*$ adjacent to a vertex $x_i$ with $ i \in [d-1]$. We claim that $w$ is not adjacent to any other vertex in $G$.

We go over possibilities of edges from $w$ to other vertices.
\begin{description}
	\item[Case 1:]  Suppose $w$ is adjacent to some vertex $x_j$ with $j \in \{0, \dotsc , d\}$. If $j=i+1$ or $j=i-1$, then $wx_ix_j$ forms a triangle $J'$. Since the path $P' = x_i,\ldots,x_{d-1},v$ that has length smaller than $P^*$ connects $J'$ and $T^*$, we have that $d_G(J', T^*) < d_G(J^*, T^*)$. This contradicts the fact that $(J^*, T^*)$ is a closest forbidden pair.

Hence $w$ is not adjacent to $x_{i-1}$ and $x_{i+1}$. Suppose $j=i+2$ or $j=i-2$. Then the graph induced by the vertices $\{w, x_i, x_{i+1}, x_{i+2}\}$ or $\{w, x_i, x_{i-1}, x_{i-2}\}$ forms the graph $T'$ which is a $C_4$. Since the path $P' = u,\ldots,x_{i}$ that has length smaller than $P^*$ connects $J^*$ and $T'$, we have that $d_G(J^*, T') < d_G(J^*, T^*)$. This contradicts the fact that $(J^*, T^*)$ is a closest forbidden pair.

Suppose now $w$ is adjacent to $x_j$ with  $0 \leq j<i-2$ or $i+2 < j \leq d$. Then note that we have a path from $x_j$ to $x_i$ of length two via $w$ which is shorter than the path $x_i \dotsc x_j$ via $P^*$. This creates a path $P'$ which is one of $u, x_1, \dotsc , x_j, w, x_i , \dotsc , x_d, v$ or $u, x_1, \dotsc , x_i, w, x_j, \dotsc , x_d, v$ from $J^*$ to $T^*$ of length smaller than $P^*$ contradicting that $P^*$ is the shortest such path. Hence $w$ is not adjacent to any of the vertices in $P^*$.

	\item[Case 2:] Suppose $w$ adjacent to a vertex $u' \in J^*$. We assume without loss of generality that among all neighbors of $w$ in $J^*$, $u'$ is the vertex that is closest to $u$ in $J^*$, i.e, $d_{J^*}(u,u')$ is minimum. We have $d_{J^*}(u,u') \leq 4$ as $J^*$ is either a long-claw or a triangle. If $3 \leq i \leq d-1$, we have a path from $u'$ to $x_i$ of length $2$ via $w$ which is shorter than the path from $u$ to $x_i$ via $P^*$. This contradicts that $P^*$ is the shortest path from $J^*$ to $T^*$. Hence $i=1$ or $i=2$. 
Let $P'$ denote the path between $u$ and $u'$ in $J^*$ and $P''$ be the subpath of $P^*$ from $u$ to $x_i$. Then $u, P'' x_i w u' P'u$ forms a cycle $C_j$ with $4 \leq j \leq 8$ without any chords. This either contradicts that $(J^*, T^*)$ is a closest forbidden pair or $G$ is $\GG_1$-free. Hence $w$ is not adjacent to any of the vertices in $J^*$.

	\item[Case 3:]   Suppose $w$ adjacent to a vertex $v' \in T^*$. We have $d_{T^*}(v,v') \leq 2$ as $T^*$ is a $C_4$. If $1 \leq i \leq d-3$, we have a path from $v'$ to $x_i$ of length $2$ via $w$ which is shorter than the path from $v$ to $x_i$ via $P^*$. This contradicts that $P^*$ is the shortest path from $J^*$ to $T^*$. Hence $d-2 \leq i \leq d-1$. Let $P'$ be the subpath of $P^*$ from $x_i$ to $v$. Then $v'wx_i P' vv'$ forms a cycle $C_j$ with $4 \leq j \leq 6$ without any chords. This either contradicts that $(J^*, T^*)$ is a closest forbidden pair or $G$ is $\GG_1$-free.

Hence $w$ is adjacent to none of the other vertices of $J^* \cup T^* \cup P^*$.

	\item[Case 4:]  We now prove that $w$ is not adjacent to any other vertex $w' \in V(G) - (J^* \cup T^* \cup P^*)$. Suppose that there exists such a vertex $w'$. Suppose $w'$ is adjacent to vertex $x_i$. Then the graph induced by the set of vertices $J' = \{w',w, x_i\}$ forms a triangle.  Since the path $P' = x_i,\ldots,x_{d-1},v$ that has length smaller than $P^*$ connects $J'$ and $T^*$, we have that $d_G(J', T^*) < d_G(J^*, T^*)$. This contradicts the fact that $(J^*, T^*)$ is a closest forbidden pair. 

	\item[- Case 4.1:]   Suppose $w'$ is adjacent to $x_{i+1}$ or $x_{i-1}$. Then the graph $T'$ induced by the vertices $\{w', w, x_i, x_{i+1}\}$ or $\{w', w, x_i, x_{i-1}\}$ forms a $C_4$. Since the path $P' = u,\ldots,x_{i}$ that has length smaller than $P^*$ connects $J^*$ and $T'$, we have that $d_G(J^*, T') < d_G(J^*, T^*)$. This contradicts the fact that $(J^*, T^*)$ is a closest forbiddden pair.

Hence this is not the case. Now suppose $w'$ is adjacent to $x_{i+2}$ or $x_{i-2}$. Then the graph induced by the set of vertices $\{w', w, x_i, x_{i+1}, x_{i+2}\}$ or $\{w', w, x_i, x_{i-1}, x_{i-2}\}$ is $C_5$ contradicting that $G$ is $\GG_1$-free. Hence this is not the case. Now suppose $w'$ is adjacent to $x_{i+3}$ or $x_{i-3}$. Then the graph induced by the set of vertices $\{w', w, x_i, x_{i+1}, x_{i+2}, x_{i+3}\}$ or $\{w', w, x_i, x_{i-1}, x_{i-2}, x_{i-3}\}$ is $C_6$ contradicting that $G$ is $\GG_1$-free. 

Now suppose $w'$ is adjacent to $x_j$ with $1 \leq j<i+3$ or $i+3 < j \leq d$. Then we have a path $P'$ which is either $u, x_1 \dotsc x_j w, x_i \dotsc x_d,v$ or $u, x_1 \dotsc x_i w, x_j \dotsc x_d, v$ from $J^*$ to $T^*$ of length smaller than $P^*$ contradicting that $P^*$ is the smallest such path. Hence $w'$ is not adjacent to any of the vertices in $P^*$. 

\item[- Case 4.2:] Suppose $w'$ adjacent to a vertex $u' \in J^*$. We assume without loss of generality that among all neighbors of $w$ in $J^*$, $u'$ is the vertex that is closest to $u$ in $J^*$, i.e, $d_{J^*}(u,u')$ is minimum. If $3 \leq i \leq d-1$, we have the path $u'w'wx_i$ from $u'$ to $x_i$ of length $3$ via $w$ which is shorter than the path from $u$ to $x_i$ via $P^*$. This contradicts that $P^*$ is the shortest path from $J^*$ to $T^*$. Hence $1 \leq i \leq 3$. Let $P'$ denote the between $u$ and $u'$ in $J^*$ and $P''$ be the prefix of the path $P^*$ from $u$ to $x_i$. Then $u P'' x_i,w,w', u' P'u$ forms a cycle $C_j$ with $4 \leq j \leq 10$ with no chords. This either contradicts that $(J^*, T^*)$ is a pair that is closest or $G$ is $\GG_1$-free. Hence $w$ is not adjacent to any of the vertices in $J^*$. 

\item[- Case 4.3:] Suppose $w'$ adjacent to a vertex $v' \in T^*$. We have $d_{T^*}(v,v') \leq 2$ as $T^*$ is a $C_4$. If $1 \leq i \leq d-4$, we have the path $v',w', w, x_i$ from $v'$ to $x_i$ of length three which is shorter than the path from $v$ to $x_i$ via $P^*$. This contradicts that $P^*$ is the shortest path from $J^*$ to $T^*$. Hence $d-3 \leq i \leq d-1$. Let $P'$ be the suffix of the path $P^*$ from $x_i$ to $v$. Then $v',w',w,x_i P' v,v'$ forms a cycle $C_j$ with $4 \leq j \leq 7$. This either contradicts that $(J^*, T^*)$ is a pair that is closest or $G$ is $\GG_1$-free.

\end{description}

Hence we conclude that $w'$ is not adjacent to any of the vertices in $J^* \cup T^* \cup P^*$. Now, we look at the induced subgraph formed by the set of vertices $J'$ which is
\begin{itemize}
\item $\{w', w, x_i, x_{i-1}, x_{i-2}, x_{i+1}, x_{i+2}\}$ for $3 \leq i \leq x_{d-2}$,
\item $\{w', w, x_i, x_{i-1}, u, x_{i+1}, x_{i+2}\}$ when $i=2$,
\item $\{w', w, x_i, u, u', x_{i+1}, x_{i+2}\}$ when $i=1$ for $u \in J^* \cap N(u)$ or 
\item $\{w', w, x_i, x_{i-1}, x_{i-2}, v, v'\}$ when $i=d-1$ for $v' \in J^* \cap N(v)$.
\end{itemize}
 This graph forms a long-claw. Since the path $P'$ from $J'$ to $v \in T^*$ has length smaller than $P^*$, the distance $d_G(J', T^*) < d_G(J^*, T^*)$. This contradicts the fact that $(J^*, T^*)$ is a closest forbidden pair. 

Hence no such vertex $w'$ exist and therefore $w$ has no other neighbors in $G$. 
\end{proof}

We now use Lemma \ref{lemma:nearest-claw-triangle-C4-pair-path-neighbors} to prove that Condition \ref{2problem-condition-4} is satisfied for {\ChordalBipPer}.

\begin{lemma}
\label{lemma:branching-rule-2-safeness}
 Condition \ref{2problem-condition-4} is satisfied for {\ChordalBipPer}.
\end{lemma}

\begin{proof}
Condition \ref{2problem-condition-4} is not satisfied in the following case. There exist a pair  $(J^*, T^*)$ which is the vertex subsets of a closest long-claw, triangle pair in a connected component of $G$, where $J^*$ one of long-claw or a triangle and $T^*$ is a $C_4$. Also there is a shortest path $P^* := x_0,x_1,\ldots,x_{d-1},x_d$ between $J^*$ and $T^*$ of length $d_G(J^*, T^*) = d$ with $x_0 = u \in J^*$ and $x_d = v \in T^*$. A forbidden set $Q$ of the graph $G$ is such that $Q$ contains some internal vertex $x_i$ of the path $P$ but it does not contain $v$.

Since the graph $G$ is $\GG_1$-free, $G[Q]$ can be a hole of size at least $11$ or contain a forbidden pair which is $($long-claw, $C_4)$ or $(C_3, C_4)$. Note that all of these possibilities contain cycles.  But from Lemma \ref{lemma:nearest-claw-triangle-C4-pair-path-neighbors}, the component $C$ of $G \setminus (J^* \cup  T^*)$ that contains the internal vertices of $P^*$ is a caterpillar which does not contain any cycles. Hence $Q$ is not fully contained in $C$.

From Lemma \ref{lemma:nearest-claw-triangle-C4-pair-path-neighbors}, the only neighbors of $C$ is $u$ and $v$ via $x_1$ and $x_{d-1}$ respectively. Since $G[Q]$ is connected and intersects $x_i$, $Q \cap \{u, v\} \neq \emptyset$. Since we assumed that $Q$ does not contain the vertex $v$, we have $u \in Q$. In particular, $Q$ contains the entire subpath of $P^*$ from $u$ to $x_i$.


Since $x_i$ does not have any paths to the vertex $u$ other than the subpath in $P^*$, the forbidden set $Q$ cannot be a hole. Hence $Q$ can only correspond to a (long-claw, $C_4$) or $(C_3, C_4)$ forbidden pair. In this case, we claim that the set after removing the vertices $x_1, \dotsc, x_i$ from $Q$ is also a forbidden set. This contradicts that $Q$ is a forbidden set as by definition they are required to be minimal.

From Lemma \ref{lemma:nearest-long-claw-triangle-pair-path-neighbors}, we can conclude that none of the vertices $x_1, \dotsc , x_i$ can belong to a subset of vertices in the graph such that the graph induced by the subset is a long-claw, triangle or a $C_4$. This is because otherwise, we get a forbidden pair using this subset which is closer than the closest forbidden pair $(J^*, T^*)$.
Since $v$ disconnects the path $P^*$, the subpath $x_1, \dotsc , x_i$ is not part of a path connecting a forbidden pair either as if so $Q$ must contain the entire path $P^*$ including $v$. Hence the set after removing the vertices $x_1, \dotsc , x_i$ from $Q$ is still a forbidden set contradicting that $Q$ is minimal.

These cases of $Q$ are mutually exhaustive completing the proof of the lemma.
\end{proof}



We have $f(k) = \max \{k^{O(k)}, 9^k\} = k^{O(k)}$ and $c = 11$. We know that {\sc Chordal Vertex Deletion} has an $\log^2(OPT)$-approximation algorithm \cite{agrawal2020polylogarithmic} where $OPT$ denote the size of the optimal solution. Also, {\sc Bipartite Permutation Vertex Deletion} has a $9$-approximation algorithm \cite{bozyk2020vertex}. Hence $d=\max \{c, c_1, c_2\} = \max \{11,\log^2(OPT),2\} = \log^2(OPT)$. We have the following theorem.

\begin{theorem}
\label{theorem:chordal-bipper-FPT-algorithm}
{\ChordalBipPer} can be solved in $k^{O(k)}poly(n)$-time and has as a $\log^2(OPT)$-approximation algorithm..
\end{theorem}



\section{Conclusion}
We gave faster algorithms for some vertex deletion problems to pairs of scattered graph classes with infinite forbidden families.
The existence of a polynomial kernel for all the problems studied except the case when both the forbidden families are finite and one of them has a path are open. It is even open when the two scattered graph classes have finite forbidden families (without the forbidden path assumption).

Currently we do not know any $W[1]$-hardness results (where we do not expect to have \FPT algorithms) for scattered graph classes in the literature when the deletion to each underlying graph classes is {\FPT}. Finding such a result would be interesting.

Another open problem is to give faster FPT algorithms for problems that doesn't fit in any of the frameworks described above, especially problems which does not have a constant sized forbidden pair family. An example is the case when $(\Pi_1, \Pi_2)$ is (Chordal, Bipartite). The forbidden pair family for this problem is the set of all pairs $(C_{2i}, C_3)$ with $i \geq 2$ which is not of constant size.

%
%
%
\bibliographystyle{plain}

\begin{thebibliography}{10}

\bibitem{agrawal2020polylogarithmic}
Akanksha Agrawal, Daniel Lokshtanov, Pranabendu Misra, Saket Saurabh, and
  Meirav Zehavi.
\newblock Polylogarithmic approximation algorithms for weighted-f-deletion
  problems.
\newblock {\em ACM Transactions on Algorithms (TALG)}, 16(4):1--38, 2020.

\bibitem{aoike2022improved}
Yuuki Aoike, Tatsuya Gima, Tesshu Hanaka, Masashi Kiyomi, Yasuaki Kobayashi,
  Yusuke Kobayashi, Kazuhiro Kurita, and Yota Otachi.
\newblock An improved deterministic parameterized algorithm for cactus vertex
  deletion.
\newblock {\em Theory of Computing Systems}, 66(2):502--515, 2022.

\bibitem{bafna19992}
Vineet Bafna, Piotr Berman, and Toshihiro Fujito.
\newblock A 2-approximation algorithm for the undirected feedback vertex set
  problem.
\newblock {\em SIAM Journal on Discrete Mathematics}, 12(3):289--297, 1999.

\bibitem{bozyk2020vertex}
{\L}ukasz Bo{\.z}yk, Jan Derbisz, Tomasz Krawczyk, Jana Novotn{\'a}, and
  Karolina Okrasa.
\newblock Vertex deletion into bipartite permutation graphs.
\newblock In {\em 15th International Symposium on Parameterized and Exact
  Computation (IPEC 2020)}. Schloss Dagstuhl-Leibniz-Zentrum f{\"u}r
  Informatik, 2020.

\bibitem{brandstadt1999graph}
Andreas Brandstadt, Jeremy~P Spinrad, et~al.
\newblock {\em Graph classes: a survey}, volume~3.
\newblock SIAM, 1999.

\bibitem{bruckner2015graph}
Sharon Bruckner, Falk H{\"u}ffner, and Christian Komusiewicz.
\newblock A graph modification approach for finding core--periphery structures
  in protein interaction networks.
\newblock {\em Algorithms for Molecular Biology}, 10(1):1--13, 2015.

\bibitem{Cai96}
Leizhen Cai.
\newblock {Fixed-Parameter Tractability of Graph Modification Problems for
  Hereditary Properties}.
\newblock {\em Inf. Process. Lett.}, 58(4):171--176, 1996.

\bibitem{cao2016linear}
Yixin Cao.
\newblock Linear recognition of almost interval graphs.
\newblock In {\em Proceedings of the Twenty-Seventh Annual ACM-SIAM Symposium
  on Discrete Algorithms}, pages 1096--1115. SIAM, 2016.

\bibitem{cao2015interval}
Yixin Cao and D{\'a}niel Marx.
\newblock Interval deletion is fixed-parameter tractable.
\newblock {\em ACM Transactions on Algorithms (TALG)}, 11(3):1--35, 2015.

\bibitem{cao2016chordal}
Yixin Cao and D{\'a}niel Marx.
\newblock Chordal editing is fixed-parameter tractable.
\newblock {\em Algorithmica}, 75(1):118--137, 2016.

\bibitem{chen2010improved}
Jianer Chen, Iyad~A Kanj, and Ge~Xia.
\newblock Improved upper bounds for vertex cover.
\newblock {\em Theoretical Computer Science}, 411(40-42):3736--3756, 2010.

\bibitem{cygan2015parameterized}
Marek Cygan, Fedor~V Fomin, {\L}ukasz Kowalik, Daniel Lokshtanov, D{\'a}niel
  Marx, Marcin Pilipczuk, Micha{\l} Pilipczuk, and Saket Saurabh.
\newblock {\em Parameterized algorithms}, volume~3.
\newblock Springer, 2015.

\bibitem{cygan2013split}
Marek Cygan and Marcin Pilipczuk.
\newblock Split vertex deletion meets vertex cover: new fixed-parameter and
  exact exponential-time algorithms.
\newblock {\em Information Processing Letters}, 113(5-6):179--182, 2013.

\bibitem{DiestelBook2012}
Reinhard Diestel.
\newblock {\em Graph Theory, 4th Edition}, volume 173 of {\em Graduate texts in
  mathematics}.
\newblock Springer, 2012.

\bibitem{fomin2012planar}
Fedor~V Fomin, Daniel Lokshtanov, Neeldhara Misra, and Saket Saurabh.
\newblock Planar f-deletion: Approximation, kernelization and optimal fpt
  algorithms.
\newblock In {\em 2012 IEEE 53rd Annual Symposium on Foundations of Computer
  Science}, pages 470--479. IEEE, 2012.

\bibitem{GanianRS17}
Robert Ganian, M.~S. Ramanujan, and Stefan Szeider.
\newblock Discovering archipelagos of tractability for constraint satisfaction
  and counting.
\newblock {\em {ACM} Trans. Algorithms}, 13(2):29:1--29:32, 2017.

\bibitem{garg1996approximate}
Naveen Garg, Vijay~V Vazirani, and Mihalis Yannakakis.
\newblock Approximate max-flow min-(multi) cut theorems and their applications.
\newblock {\em SIAM Journal on Computing}, 25(2):235--251, 1996.

\bibitem{jacob2021parameterized}
Ashwin Jacob, Jari~JH de~Kroon, Diptapriyo Majumdar, and Venkatesh Raman.
\newblock Deletion to scattered graph classes i - case of finite number of
  graph classes.
\newblock {\em arXiv preprint arXiv:2105.04660}, 2021.

\bibitem{jacob2020parameterized}
Ashwin Jacob, Diptapriyo Majumdar, and Venkatesh Raman.
\newblock Parameterized complexity of deletion to scattered graph classes.
\newblock In {\em 15th International Symposium on Parameterized and Exact
  Computation (IPEC 2020)}. Schloss Dagstuhl-Leibniz-Zentrum f{\"u}r
  Informatik, 2020.

\bibitem{JacobMR21fct}
Ashwin Jacob, Diptapriyo Majumdar, and Venkatesh Raman.
\newblock Faster {FPT} algorithms for deletion to pairs of graph classes.
\newblock In Evripidis Bampis and Aris Pagourtzis, editors, {\em Fundamentals
  of Computation Theory - 23rd International Symposium, {FCT} 2021, Athens,
  Greece, September 12-15, 2021, Proceedings}, volume 12867 of {\em Lecture
  Notes in Computer Science}, pages 314--326. Springer, 2021.

\bibitem{jansen2014near}
Bart~MP Jansen, Daniel Lokshtanov, and Saket Saurabh.
\newblock A near-optimal planarization algorithm.
\newblock In {\em Proceedings of the Twenty-Fifth Annual ACM-SIAM Symposium on
  Discrete Algorithms}, pages 1802--1811. SIAM, 2014.

\bibitem{kawarabayashi2017polylogarithmic}
Ken-ichi Kawarabayashi and Anastasios Sidiropoulos.
\newblock Polylogarithmic approximation for minimum planarization (almost).
\newblock In {\em 2017 IEEE 58th Annual Symposium on Foundations of Computer
  Science (FOCS)}, pages 779--788. IEEE, 2017.

\bibitem{KP14}
Tomasz Kociumaka and Marcin Pilipczuk.
\newblock {Faster deterministic Feedback Vertex Set}.
\newblock {\em Inf. Process. Lett.}, 114(10):556--560, 2014.

\bibitem{lekkeikerker1962representation}
C~Lekkeikerker and J~Boland.
\newblock Representation of a finite graph by a set of intervals on the real
  line.
\newblock {\em Fundamenta Mathematicae}, 51(1):45--64, 1962.

\bibitem{LewisY80}
John~M. Lewis and Mihalis Yannakakis.
\newblock The node-deletion problem for hereditary properties is np-complete.
\newblock {\em J. Comput. Syst. Sci.}, 20(2):219--230, 1980.

\bibitem{lokshtanov2014faster}
Daniel Lokshtanov, NS~Narayanaswamy, Venkatesh Raman, MS~Ramanujan, and Saket
  Saurabh.
\newblock Faster parameterized algorithms using linear programming.
\newblock {\em ACM Transactions on Algorithms (TALG)}, 11(2):1--31, 2014.

\bibitem{van2013proper}
Pim Van’t~Hof and Yngve Villanger.
\newblock Proper interval vertex deletion.
\newblock {\em Algorithmica}, 65(4):845--867, 2013.

\end{thebibliography}
%


\end{document}